\documentclass[journal]{IEEEtran}
\twocolumn \def\twocolumnmode{} 

\usepackage{cite} 
\usepackage{amsmath}
\usepackage{amsfonts}
\usepackage{amssymb}
\usepackage{graphicx}
\newtheorem{theorem}{Theorem}

\newtheorem{definition}[theorem]{Definition}

\newtheorem{lemma}[theorem]{Lemma}

\newenvironment{proof}[1][Proof]{\textbf{#1.} }{\ \rule{0.5em}{0.5em}}

\usepackage{pdflscape}
\usepackage{multirow}

\usepackage{pgf}
\usepackage{tikz}
\usetikzlibrary{arrows,automata}
\usepackage[latin1]{inputenc}
\usepackage{verbatim}

\newcommand{\Ts}{T_s}  
\newcommand{\T}{t} 

\newcommand{\K}{K}
\newcommand{\idx}{f} 		
\newcommand{\blk}{n}

\newcommand{\Y}{\textbf{Y}}

\newcommand{\y}{\textbf{y}}

\newcommand{\snr}{{P/N}}
\newcommand{\Pu}{\textsf{P}}

\newcommand{\Yv}[1]{\textbf{Y}_{#1}}
\newcommand{\yv}[1]{\textbf{y}_{#1}}

\usepackage{scalefnt}

\ifCLASSINFOpdf
\else
\fi

\begin{document}
\title{Models and Information Rates for Multiuser Optical Fiber Channels with Nonlinearity and Dispersion}

\newcommand{\codered}[1]{{\color[rgb]{1,0,0} {#1}}}

\author{Hassan~Ghozlan,~\IEEEmembership{Member,~IEEE},
        and~Gerhard~Kramer,~\IEEEmembership{Fellow,~IEEE}
\thanks{
The work of H. Ghozlan was supported by a USC Annenberg Fellowship and 
the National Science Foundation (NSF) under Grant CCF-09-05235.
The work of G. Kramer was supported by an Alexander von Humboldt Professorship endowed by
the German Federal Ministry of Education and Research, as well as by the NSF under Grant CCF-09-05235.
Part of the material in this paper was presented
at the IEEE International Symposium on Information Theory, Austin, TX, June, 2010 and
at the IEEE International Symposium on Information Theory, Saint Petersburg, Russia, July/August, 2011.}
\thanks{H. Ghozlan was with the Department
of Electrical Engineering, University of Southern California, Los Angeles,
CA 90089, USA.
He is now with Intel Corporation, Hillsboro, OR 97124, USA.
(e-mail: hassan.ghozlan@intel.com).}
\thanks{G. Kramer is with the Institute for Communications Engineering, 
Technical University of Munich, 80333 Munich, Germany (email: gerhard.kramer@tum.de).}
}

\markboth{Journal,~Vol.~, No.~, Month~Year}
{Shell \MakeLowercase{\textit{et al.}}: Bare Demo of IEEEtran.cls for Journals}

\maketitle

\IEEEpeerreviewmaketitle

\begin{abstract}
Two discrete-time interference channel models are developed 
for information transmission over a single span of optical fiber using wavelength-division multiplexing (WDM) and lumped amplification. 
The models are derived from the nonlinear Schr\"odinger (NLS) equation 
by including the nonlinear phenomena of self-phase modulation (SPM) and cross-phase modulation (XPM)
but ignoring four-wave mixing (FWM), polarization effects and group velocity dispersion (GVD) within WDM bands.
The first model also ignores GVD across WDM bands, referred to as group velocity mistmatch (GVM).
For the case of two users, 
a new technique called \textit{interference focusing} is proposed where each carrier achieves the capacity pre-log 1, 
thereby doubling the pre-log of 1/2 achieved by using conventional methods. 
For three users, interference focusing is also useful under certain conditions.
The second model captures GVM and the effect of filtering at the receivers in addition to SPM and XPM.
In a 3-user system, it is shown that all users can achieve the maximum pre-log factor 1 simultaneously
by using interference focusing, a time-limited pulse and a bank of filters at the receivers.
\end{abstract}

\begin{IEEEkeywords}
Optical fiber, wavelength-division multiplexing, 
Kerr nonlinearity, cross-phase modulation, 
chromatic dispersion, group velocity mismatch,
interference channel.
\end{IEEEkeywords}

\section{Introduction}
\newcommand{\alert}[1]{#1}
\newcommand{\SPM}{\text{SPM}} 
\newcommand{\XPM}{\text{XPM}}	
\newcommand{\FWM}{\text{FWM}}	
\newcommand{\ASE}{\text{ASE}}	
\newcommand{\eff}{\text{eff}}
\newcommand{\scripttwo}[2]{ \scriptsize{\begin{array}{c} #1 \\ #2  \end{array}} }
\newcommand{\NLI}{\text{NLI}}	

\newcommand{\tnr}[1]{{\textnormal{#1}}}
\newcommand{\ie}{i.e.,~}
\newcommand{\E}{\mathbb{E}}
\newcommand{\PASE}{P_\text{ASE}}
\newcommand{\Zn}{Z_{k}}\newcommand{\zn}{z_{k}}
\newcommand{\Znt}{\tilde{Z}_{k}}
\newcommand{\Ex}{\mathbb{E}}

\IEEEPARstart{T}{he} majority of traffic in core networks is carried by optical fiber.
Understanding the ultimate limits of communication over optical fiber is thus of great importance
and would help to provide guidelines for designing networks.
An appealing property of fiber is that it has low attenuation over a large range of frequencies
which allows the transmission of broadband signals over long distances.
Optical amplifiers compensate the power loss
but they add noise.
Moreover, a signal propagating in fiber experiences distortions due to chromatic dispersion and Kerr nonlinearity. 
The fiber channel thus suffers from three main impairments of different nature: noise, dispersion, and Kerr nonlinearity. 
The interaction between these three phenomena makes the problem of estimating the capacity challenging \cite{JLT2010}.

\subsection{Capacity Estimates}
There are many approaches to estimate the capacity of optical fiber channels. 
The technical papers fall into two main categories: they either study the capacity  of simplified models, 
or they develop capacity lower bounds (achievable rates) on the full model by simulation. 
We next review these papers. Our document belongs to the former category.

Splett et al \cite{Splett1993} study a single-channel system and 
derive an approximate formula for the power spectral density of the intrachannel four-wave mixing (FWM) at the center frequency 
assuming the input signal has uncorrelated spectral components. 
They derive an achievable information rate expression by treating FWM as additive Gaussian noise. 
The information rate has a peak at a finite input power. 
They modify the power spectral density expression of FWM to obtain a similar result for multi-channel systems 
where cross-phase modulation (XPM) is ignored.
Narimanov and Mitra \cite{Narimanov2002}
study a single-channel transmission over a multi-span dispersive fiber link.
They use a perturbation technique to approximate the solution to the nonlinear Schr\"{o}dinger (NLS) equation assuming that the nonlinear term is small
and they derive a capacity expression.
Xiang and Zhang \cite{Xiang2011Perturb} extend some of the results of \cite{Narimanov2002}.

Mecozzi \cite{Mecozzi94} models the propagation of a single signal in a dispersionless fiber link, 
in which the fiber loss is compensated by distributed amplification.
Mecozzi derives an expression for the conditional distribution of the output field given the input field
by computing all (conditional) moments.
Turitsyn et al \cite{TuritsynPRL2003} also
study single-channel transmission 
over zero-dispersion fiber links.
They obtain the conditional distribution 
using techniques from quantum mechanics.
For Gaussian inputs, a sampling receiver and direct-detection,
a lower bound on capacity is derived
that grows logarithmically with the signal-to-noise ratio (SNR) with a pre-log = 1/2.
In \cite{Yousefi2010SumProd,Yousefi2010FokkerPlank,PerSampleIT2011},
Yousefi and Kschischang derive the conditional probability using two different approaches:
a sum-product approach and a Fokker-Planck differential equation approach.
Wei and Plant \cite{WeiPlantArxiv2006} make useful comments on the results of \cite{TuritsynPRL2003},\cite{Mitra2001} and \cite{Tang2001}.

Djordjevic et al \cite{DjordjevicJLT2005} 
study a single-channel system and
estimate numerically the achievable information rate for independent uniformly distributed inputs
when the intrachannel Kerr nonlinearity, chromatic dispersion and amplified spontaneous emission 
are taken into account.
They use a finite-state machine approach 
where the state is determined by a number of past and future inputs surrounding the current input,
and the conditional distribution of the output given the state
is approximated using histograms.
Ivakovic et al \cite{IvakovicJLT2007} follow \cite{DjordjevicJLT2005} and
propose an approximate expression for the conditional output distribution
when on-off keying (OOK) is used to circumvent the computation of histograms.
These methods are limited to low-order modulation for complexity reasons.

Mitra and Stark \cite{Mitra2001} 
study a wavelength division multiplexing (WDM) system in which XPM is the only nonlinear effect, i.e.,
they ignore FWM and assume that self-phase modulation (SPM) can be fully corrected.
A key simplification in \cite{Mitra2001} is approximating the sum of intensities of the interfering channels in the XPM term of the propagation equation by a Gaussian random process.
A lower bound on capacity (per WDM channel) is derived for Gaussian inputs
using the input-output covariance matrix.
The conclusion of \cite{Mitra2001} is that the lower bound has a peak and does not increase indefinitely with the input power.
Wegener et al \cite{Wegener2004}
also study WDM transmission over a multi-span dispersive fiber link.
To simplify the solution of the coupled propagation equations analytically,
the technique of \cite{Mitra2001} is used and the FWM is replaced with a Gaussian random process.
A lower bound on capacity is evaluated using the input-output covariance matrix.

Ho and Kahn \cite{HoKahnOFC2002}
study WDM transmission over a multi-span dispersive fiber link.
They argue that under constant-envelope (or constant-intensity) modulation
with uniform phase\footnote{We refer to constant-envelope modulation with uniform phase as \emph{ring} modulation.},
SPM and XPM cause only time-invariant phase shifts 
and hence the phase distortion is eliminated.
By modeling FWM as additive Gaussian noise, 
they obtain an estimate of the information rate achieved by constant-envelope modulation.
The FWM components from individual fiber spans are assumed to combine incoherently.

Tang \cite{Tang2001} studies WDM transmission over a single-span dispersion-free fiber link.
In this case, the propagation equation can be solved analytically in closed-form.
A lower bound on capacity is obtained for Gaussian inputs by computing the power spectral density of the input (the sum over all WDM channels), the power spectral density of the output (the overall WDM signal after propagation) and the cross-spectral density of the input and output.
Tang extends the results of \cite{Tang2001} to a multi-span dispersion-free fiber link in \cite{TangMultispan2001}
and then to a multi-span dispersive fiber link in \cite{Tang2002}.
In \cite{Tang2002}, a truncated Volterra series \cite{Ped1997Volterra} 
is used to approximate the solution to the NLS equation 
assuming that the effect of nonlinearity is small.
The lower bounds in \cite{Tang2001}, \cite{TangMultispan2001} and \cite{Tang2002} have a peak value at finite input powers.

Taghavi et al \cite{Taghavi2006} study WDM transmission over a single-span dispersive fiber link.
They use a (truncated) Volterra series solution to 
the propagation equation.
Each receiver uses a linear filter to compensate dispersion followed by a matched filter (matched to the transmitted pulse) whose output is sampled at the symbol rate.
Assuming that dispersion is weak (so that inter-symbol interference can be neglected),
a discrete-time memoryless model is obtained.
Each receiver has access to the received signal of all channels 
and thus this case is treated as a multiple-access channel.
It is found that nonlinearity does not affect the capacity to the first-order approximation (in the nonlinear coefficient)
and high rates are achieved by performing interference cancellation 
before decoding.
Moreover, single-channel detection (i.e., the decoder for a given user has access to the received signal at its own wavelength only)
is considered in two regimes: XPM-dominated and FWM-dominated regimes.
The capacity for single-channel detection is significantly reduced compared to the multiple access channel capacity.

Essiambre et al \cite{JLT2010}
review fundamental concepts of digital communications, information theory
and the physical phenomena present in transmission over optical fiber networks.
They estimate by numerical simulations capacity lower bounds for WDM using multi-ring constellations, 
different constellation shapings and different fiber dispersion maps. 
Nonlinear compensation through backpropagation of individual channels is used.
The trend in the various scenarios is that the capacity lower bound has a peak value at a finite launch power.

Bosco et al \cite{Bosco2011GN,Bosco2012GN:erratum}
study WDM transmission over uncompensated optical fiber links
with both distributed and lumped amplification.
They argue that, after digital signal processing (DSP) at the receiver,
the distribution of each of the received constellation points is approximately Gaussian 
with independent components, even in the absence of additive ASE noise. 
Hence, they adopt a model, called the Gaussian noise (GN) model, 
in which the impact of nonlinear propagation is approximated by excess additive Gaussian noise (see also \cite{Splett1993}).
Using the GN model, capacity estimates are derived.
In \cite{PoggioliniJLT2014GN}, Poggiolini discusses the GN model in depth.

Mecozzi and Essiambre \cite{MecozziJLT2012}
study multi-channel transmission over a dispersive fiber link with distributed amplification. 
They develop a first-order perturbation theory of the signal propagation
and simplify the expression for highly dispersive, or pseudolinear, transmission. 
The signal is linearly-modulated\footnote{The signal is the sum of modulated pulses.} at the transmitter 
and the detection apparatus at the receiver is made of an optical filter to separate the channel, mixing with a local oscillator 
and subsequent sampling at the symbol rate. 
By concentrating on inter-channel nonlinearity, in particular XPM, 
they derive a capacity estimate per channel.
An important observation is that
the kurtosis of
the constellation of the interfering channels is important in determining the system impairments. 

Secondini et al \cite{SecondiniJLT2013}
study WDM transmission over a dispersive fiber link.
FWM is neglected.
The key simplification is 
replacing the unknown intensities appearing in the propagation equation with those corresponding to linear propagation.
They then derive a first-order approximation to the solution
based on frequency-resolved logarithmic perturbations.
The approximate solution is used to develop a linear time-varying discrete-time model for the channel which is composed of the
optical fiber link followed by a back-propagation block (and thus it is assumed that SPM is fully compensated), a matched filter, 
and sampling at the symbol rate.
By using the theory of mismatched decoding, they compute 
the information rate achieved by independent and identically distributed (i.i.d.) Gaussian input symbols and a maximum likelihood \emph{symbol-by-symbol} detector 
designed for a memoryless additive white Gaussian noise (AWGN) auxiliary channel with the same covariance matrix as the true channel.
They also evaluate the information rate achieved by a maximum likelihood \emph{sequence} decoder 
designed for an auxiliary AWGN channel with inter-symbol interference, and with the same input-output covariance matrix as the true channel.

Dar et al \cite{DarAllerton2013} propose a block-memoryless discrete-time channel model 
for WDM transmission in the pseudo-linear regime in which XPM is the dominant nonlinear effect.
The model is a discrete-time phase noise channel
in which the phase noise process
models XPM  and is assumed to be a block-independent process, i.e., 
it remains unchanged within a block but changes independently between blocks.
It is assumed that the phase noise is (real) Gaussian with zero mean and a variance
that depends on the type of modulation.
For the proposed model, two lower bounds on capacity are developed: 
the first is tight in the low power regime 
while the second is better at high power.
In \cite{DarLetters2014, DarECOC2013}, Dar et al add an extra term 
to capture nonlinear effects that do not manifest themselves as phase noise.
Agrell et al \cite{AgrellJLT2014}
propose a discrete-time model called the \emph{finite-memory GN model}
for coherent long-haul fiber links without dispersion compensation.
Using the finite-memory GN model, they derive semi-analytic lower bounds for non i.i.d. inputs.
Numerical simulations show that the information rates 
of the finite-memory GN model 
are higher than the rates of the regular GN model.
We \alert{remark} that the proposed discrete-time model
is not derived from a continuous-time description of the system.

Yousefi and Kschischang
\cite{YousefiISIT2013,YousefiISIT2013Fiber,YousefiNFT1,YousefiNFT2,YousefiNFT3}
discuss the nonlinear Fourier transform (NFT), a method for 
solving a broad class of nonlinear differential equations,
and in particular for solving the NLS equation
for noiseless propagation.
They propose a scheme, 
called nonlinear frequency-division multiplexing (NFDM), 
which can be viewed as a nonlinear analogue of orthogonal frequency-division multiplexing (OFDM).
In NFDM, information is encoded in the 
NFT of the signal consisting of two components: a discrete and a continuous spectral function. 
By modulating non-interacting degrees of freedom of a signal, 
deterministic crosstalk between signal components due to dispersion and nonlinearity is eliminated,
i.e., inter-symbol and inter-channel interference are zero
if there is no noise.

\subsection{Contributions and Organization}

We develop discrete-time interference channel models for WDM transmission over a single span of both dispersionless and dispersive fiber.
The models are based on coupled differential equations that capture SPM, XPM and group velocity mismatch (GVM). 
Transmitters send linearly-modulated pulses while receivers use 
matched filters with symbol rate sampling (for dispersionless transmission) or banks of filters (for dispersive transmission).
Rather than using Gaussian codebooks, 
we design codebooks based on a new technique called \emph{interference focusing}.
We show that all users achieve a pre-log of 1 simultaneously by using interference focusing.
This paper extends the results in \cite{GhozlanISIT2010} and \cite{GhozlanGVM2011}.
More specifically, 
we extend the two-user model with a rectangular pulse in the non-zero GVM case to
a three-user model with a general time-limited (of one symbol interval) pulse 
and we also derive a capacity outer bound.
We highlight two aspects of our work (including \cite{GhozlanISIT2010} and \cite{GhozlanGVM2011}):
\begin{itemize}
	\item
	We study an \emph{interference channel} model for multiuser communication in nonlinear optical fiber.
	In contrast, most models in the literature reduce interference to be an additional source of noise and treat the problem as a \emph{point-to-point channel}.
	
	\item 
	We derive precise discrete-time models from continuous-time models with noise and filtering.
	In contrast, many publications derive or assume simplified discrete-time models 
	based on direct sampling of the continuous-time received signals without filtering.
\end{itemize}

The paper is organized as follows.
In Sec. \ref{sec:fiber_channel}, we describe the wave propagation equation in optical fiber and the impairments that arise in transmission.
We study the case of zero group velocity mismatch (zero dispersion) in Sec. \ref{sec:zero-gvm}.
We extend this model to non-zero group velocity mismatch in Sec. \ref{sec:nonzero-gvm}.
For both cases, we develop discrete-time interference channel models 
and show that a pre-log of 1 is achievable for all users, 
despite XPM that arises due to the fiber nonlinearity.
Sec. \ref{sec:ialign} relates interference focusing to interference alignment.
Sec. \ref{sec:conclusion} concludes the paper.

\subsection{Notation}
We use common notation for probability distributions and information-theoretic quantities. 
Random variables are usually written as uppercase letters and their realizations as lowercase letters. 
Probability distributions and densities are labeled with the random variables, e.g., 
the probability density of $X$ is written as $p_X(\cdot)$ and 
the conditional probability density of $Y$ given $X$ evaluated at $Y=y$ and $X=x$ is written as $p_{Y|X}(y|x)$. 
The expectation of $X$ is denoted by $\mathbb{E}[X]$.
The expressions $H(X)$, $H(Y|X)$, $H(XY)$ represent 
the entropy of $X$, the conditional entropy of $Y$ given $X$, and the joint entropy of $XY$.
The expressions $h(X)$, $h(Y|X)$, $h(XY)$ represent differential entropies. 
The mutual information of $X$ and $Y$ is written as $I(X;Y)$, and 
the mutual information of $X$ and $Y$ conditioned on Z is written as $I(X;Y|Z)$.

\section{Fiber Models} 
\label{sec:fiber_channel}
We next discuss noise, chromatic dispersion and Kerr nonlinearity in optical fiber.
Amplifiers add noise to the signal due to amplified spontaneous emission (ASE). 
The noise is typically modeled as a white Gaussian process.
There are two types of amplification: lumped and distributed.
In lumped amplification, $N_s$ amplifiers are inserted periodically over a fiber link of total length $L$ 
which creates $N_s$ spans, often each of the same length $L_s = L/N_s$.
A commonly-used lumped amplifier is the erbium-doped optical amplifier (EDFA).
In distributed amplification, the signal is amplified continuously as it propagates through the fiber.
Distributed amplification is accomplished by using Raman pumping.
For multispan lumped or distributed amplification, 
signal-noise interaction occurs because of fiber nonlinearity. 
However, there is no signal-noise interaction in the single-span lumped amplification case,
and this is the case we consider for the rest of the paper for simplicity.
This model is sometimes used as an approximation when the noise is weak and the launch power is low.

Dispersion arises because the medium absorbs energy through the oscillations of bound electrons, causing a {\it frequency} dependence of the material refractive index~\cite[p.~7]{Agrawal}. The Kerr effect is caused by anharmonic motion of bound electrons in the presence of an intense electromagnetic field, causing an {\it intensity} dependence of the material refractive index~\cite[p.~17, 165]{Agrawal}.

Suppose an optical field propagates at a center/carrier frequency $\omega_0$.
Let $A(z,t)$ be a complex number representing the slowly-varying component (or envelope) of a linearly-polarized, electric field at position $z$ and time $t$ in single-mode fiber.
We ignore polarization effects, i.e., a linearly-polarized input electric field remains linearly polarized during propagation.
The equation governing the evolution of $A(z,t)$ as the wave propagates through the fiber is \cite[p. 44]{Agrawal}
\begin{align}
	\frac{\partial A}{\partial z} + \beta_1 \frac{\partial A}{\partial t}  + i \frac{\beta_2}{2} \frac{\partial^2 A}{\partial t^2} 
	= i \gamma |A|^2 A
	\label{eq:propagation_eq}
\end{align}
where $i = \sqrt{-1}$, 
$\beta_1$ is the reciprocal of the group velocity,
$\beta_2$ is the group velocity dispersion (GVD) coefficient, and
$\gamma$ is the nonlinear coefficient.
It is common to specify GVD through 
the dispersion parameter $D$ which is related to $\beta_2$ by \cite[p. 11]{Agrawal}
\begin{align}
D = - \frac{2 \pi c \beta_2}{\lambda_0^2}
\label{eq:d_coeff_def}
\end{align}
where $\lambda_0$ is the wavelength in free-space, i.e., $\lambda_0 = 2\pi c/\omega_0$,
and $c$ is the speed of light in free space.
By defining a retarded-time reference frame with $T=t-\beta_1 z$, we have
\begin{align}
	i \frac{\partial A}{\partial z} - \frac{\beta_2}{2} \frac{\partial^2 A}{\partial T^2} + \gamma |A|^2 A = 0
	\label{eq:nls_eq}
\end{align}
which is referred to as the nonlinear Schr\"{o}dinger (NLS) equation 
because of its similarity to the Schr\"{o}dinger equation with a nonlinear potential term 
when the roles of time and distance are exchanged \cite[p. 50]{Agrawal}.
The NLS equation has no closed-form solution for general inputs.
Closed-form solutions to the NLS equation exist when $\beta_2 = 0$ and/or $\gamma = 0$.
Solutions to the NLS equation with $\beta_2 \neq 0$ and $\gamma \neq 0$ exist only for 
special input waves called solitons.

There are other interesting cases where closed-form solutions exist.
Consider a three-channel WDM system in which
three optical fields at different center frequencies $\omega_1$, $\omega_2$ and $\omega_3$ are launched into the fiber, 
i.e., the input field is\footnote{
We ignore the frequency dependence of the modal distribution. 
The difference is small and can be neglected in practice \cite[7.1.2]{Agrawal}.
}
\begin{align}
A(0, t) = \sum_{k=1}^3 A_k(0, t) e^{-i(\omega_k-\omega_0) t}.
\end{align}
Suppose $A(z, t)$ takes the form 
\begin{align}
A(z, t) = \sum_{k=1}^3 A_k(z, t) e^{i\hat{\beta}(\omega_k)z} e^{-i(\omega_k-\omega_0) t}
\end{align}
where $\hat{\beta}(\omega) = \beta_1 (\omega-\omega_0) + (\beta_2/2) (\omega-\omega_0)^2$.
By substituting into (\ref{eq:propagation_eq}), we have\footnote{
We do not use a retarded frame because it does not lead to much simplification.
This is because it is not possible to eliminate simultaneously 
all the terms of first-order derivatives with respect to time.}
\ifdefined\twocolumnmode{
\begin{align}
	&
	\sum_{k=1}^3 e^{i\hat{\beta}(\omega_k)z} e^{-i(\omega_k-\omega_0)t}
	\Bigg[
	i \frac{\partial A_k}{\partial z} + i \beta_{1k} \frac{\partial A_k}{\partial t} 
	- \frac{\beta_{2k}}{2} \frac{\partial^2 A_k}{\partial t^2}
	\nonumber\\&\qquad
	+ \gamma_k (|A_k|^2 + 2 \sum_{k^\prime \neq k} |A_{k^\prime}|^2 ) A_k
	\Bigg]
	+ \textsf{F}
	= 0
	\label{eq:nls_wdm}
\end{align}
}\else{
\begin{align}
	\sum_{k=1}^3 e^{i\hat{\beta}(\omega_k)z} e^{-i(\omega_k-\omega_0)t}
	\Bigg[
	i \frac{\partial A_k}{\partial z} + i \beta_{1k} \frac{\partial A_k}{\partial t} 
	- \frac{\beta_{2k}}{2} \frac{\partial^2 A_k}{\partial t^2}
	+ \gamma_k (|A_k|^2 + 2 \sum_{k^\prime \neq k} |A_{k^\prime}|^2 ) A_k
	\Bigg]
	+ \textsf{F}
	= 0
	\label{eq:nls_wdm}
\end{align}
}\fi
where $\beta_{1k} = \beta_1 + \beta_2 (\omega_k-\omega_0)$, $\beta_{2k} = \beta_2$, $\gamma_{k} = \gamma$ and
\ifdefined\twocolumnmode{
\begin{align}
\textsf{F} &=
\sum_{k_1 \neq k_2, k_3 \neq k_2}
e^{i \hat{\beta}(\omega_{k_1}-\omega_{k_2}+\omega_{k_3}) z} 
e^{-i(\omega_{k_1}-\omega_{k_2}+\omega_{k_3}-\omega_0) t}
\times \nonumber \\& \qquad\qquad\qquad
\Bigg[
A_{k_1} A_{k_2}^* A_{k_3} 
e^{i \Delta\hat{\beta}(\omega_{k_1},\omega_{k_2},\omega_{k_3}) z} 
\Bigg]
\label{eq:fwm}
\end{align}
}\else{
\begin{align}
\textsf{F} =
\sum_{k_1 \neq k_2, k_3 \neq k_2}
e^{i \hat{\beta}(\omega_{k_1}-\omega_{k_2}+\omega_{k_3}) z} 
e^{-i(\omega_{k_1}-\omega_{k_2}+\omega_{k_3}-\omega_0) t}
\Bigg[
A_{k_1} A_{k_2}^* A_{k_3} 
e^{i \Delta\hat{\beta}(\omega_{k_1},\omega_{k_2},\omega_{k_3}) z} 
\Bigg]
\label{eq:fwm}
\end{align}
}\fi
with $\Delta\hat{\beta}(\omega_{k_1},\omega_{k_2},\omega_{k_3})$ defined as
\begin{align*}
\Delta\hat{\beta}
= \hat{\beta}(\omega_{k_1})-\hat{\beta}(\omega_{k_2})+\hat{\beta}(\omega_{k_3})
- \hat{\beta}(\omega_{k_1}-\omega_{k_2}+\omega_{k_3}).
\end{align*}
The summands in (7) are called FWM terms 
because they involve mixing, i.e., energy transfer, between four frequencies:
$\omega_{k_1}$, $\omega_{k_2}$, $\omega_{k_3}$ and $\omega_{k_1}-\omega_{k_2}+\omega_{k_3}$
for $k_1 \neq k_2$, $k_3 \neq k_2$.
We remark that the phase-matching condition $\Delta\hat{\beta} = 0$ should be satisfied for new frequency components to build up significantly via FWM, 
a condition not generally satisfied in practice when there is dispersion \cite[Sec. 7.1.1]{Agrawal}.

We ignore all FWM terms, i.e., we set $\textsf{F}=0$ in (\ref{eq:nls_wdm}).
Therefore, we have the coupled equations
\ifdefined\twocolumnmode{
\begin{align}
	&
	i \frac{\partial A_k}{\partial z} + i \beta_{1k} \frac{\partial A_k}{\partial t} 
	- \frac{\beta_{2k}}{2} \frac{\partial^2 A_k}{\partial t^2}
	\nonumber\\&\qquad~
	+ \gamma_k (|A_k|^2 + 2 \sum_{k^\prime \neq k} |A_{k^\prime}|^2 ) A_k
	= 0
	\label{eq:coupled_eq}
\end{align}
}\else{
\begin{align}
	i \frac{\partial A_k}{\partial z} + i \beta_{1k} \frac{\partial A_k}{\partial t} 
	- \frac{\beta_{2k}}{2} \frac{\partial^2 A_k}{\partial t^2}
	+ \gamma_k (|A_k|^2 + 2 \sum_{k^\prime \neq k} |A_{k^\prime}|^2 ) A_k
	= 0
	\label{eq:coupled_eq}
\end{align}
}\fi
for $k=1,2,3$,
assuming that the three optical fields do not overlap in the frequency domain.
There are two nonlinear terms in (\ref{eq:coupled_eq}): the first is referred to as SPM and the second term is referred to as XPM.
The term \emph{phase modulation} is because, in absence of GVD,
Kerr nonlinearity leaves the pulse shape unchanged but causes an intensity-dependent phase shift due to the signal \emph{itself} (SPM)
and \emph{co-propagating} signals (XPM).
XPM is an important impairment in optical networks using WDM, see~\cite{JLT2010}.
There are also two terms in (\ref{eq:coupled_eq}) due to dispersion. 
The first term with $\beta_{1k}$ captures the mismatch in group velocity \emph{between} channels while
the second term with $\beta_{2k}$ captures the GVD \emph{within} the bandwidth of a channel.

Similar to the NLS equation (\ref{eq:nls_eq}), 
the coupled equations in (\ref{eq:coupled_eq}) have no closed-form solution for a general input.
Therefore, we make a further simplification by ignoring the GVD within a channel, i.e.,
we set $\beta_{2k} = 0$ for $k=1,2,3$. 
This simplification gives the closed-form solution:\footnote{
The solution follows from steps
similar to the steps outlined in Sec. 1.8.10 of \cite{GhozlanPHD} for two coupled equations.
}
\begin{align}
	A_k(L,t) &= A_{k}(0,t-\beta_{k1} L) \exp\left( i \phi_{k}(L,t-\beta_{k1} L) 	\right)
	\label{eq:Ak_solution}
\end{align}
where $k\in\{1,2,3\}$, 
$L$ is the length of a single span of fiber and
the time-dependent nonlinear phase shifts $\phi_{k}(L,t)$ are 
\ifdefined\twocolumnmode{
\begin{align}
	\phi_{1}(L,t) &=
	\int_0^L \gamma_1 \big( |A_{1}(0,t)|^2  
								 +  2 |A_{2}(0,t+d_{12} \zeta)|^2 
								\nonumber\\&\qquad\qquad
								 +	2 |A_{3}(0,t+d_{13} \zeta)|^2 \big) d\zeta 
	\label{eq:phi1_solution}\\
	\phi_{2}(L,t) &=
	\int_0^L \gamma_2 \big( |A_{2}(0,t)|^2
								 + 	2 |A_{1}(0,t+d_{21} \zeta)|^2 
								\nonumber\\&\qquad\qquad
								 + 	2 |A_{3}(0,t+d_{23} \zeta)|^2 \big) d\zeta 
	\label{eq:phi2_solution}\\
	\phi_{3}(L,t) &=
	\int_0^L \gamma_3 \big( |A_{3}(0,t)|^2
								 + 	2 |A_{1}(0,t+d_{31} \zeta)|^2 
								\nonumber\\&\qquad\qquad
								 + 	2 |A_{2}(0,t+d_{32} \zeta)|^2 \big) d\zeta 								 			 
	\label{eq:phi3_solution}
\end{align}
}\else{
\begin{align}
	\phi_{1}(L,t) &=
	\int_0^L \gamma_1 \big( |A_{1}(0,t)|^2  
								 +  2 |A_{2}(0,t+d_{12} \zeta)|^2 
								 +	2 |A_{3}(0,t+d_{13} \zeta)|^2 \big) d\zeta 
	\label{eq:phi1_solution}\\
	\phi_{2}(L,t) &=
	\int_0^L \gamma_2 \big( |A_{2}(0,t)|^2
								 + 	2 |A_{1}(0,t+d_{21} \zeta)|^2 
								 + 	2 |A_{3}(0,t+d_{23} \zeta)|^2 \big) d\zeta 
	\label{eq:phi2_solution}\\
	\phi_{3}(L,t) &=
	\int_0^L \gamma_3 \big( |A_{3}(0,t)|^2
								 + 	2 |A_{1}(0,t+d_{31} \zeta)|^2 
								 + 	2 |A_{2}(0,t+d_{32} \zeta)|^2 \big) d\zeta 								 			 
	\label{eq:phi3_solution}
\end{align}
}\fi
where 
\begin{align}
	d_{kj} \stackrel{\Delta}{=} \beta_{1k}-\beta_{1j}
\end{align}
is a measure of GVM between channel $k$ and channel $j$.
We remark that our model captures GVD of the \emph{overall} signal, 
but only through GVM, 
namely through $d_{kj} = \beta_2 (\omega_k-\omega_j)$.

As we pointed out earlier,
we assume lumped amplification at the receiver,
i.e., the signal observed at receiver $k$ 
after removing the constant phase shift $e^{i\hat{\beta}(\omega_k)L}$
is
\begin{align}
	r_k(t) = A_k(L,t) + z_k(t)
	\label{eq:rx_k}
\end{align}
where $z_k(t)$ is circularly-symmetric white Gaussian noise with 
$\mathbb{E}[z_k(t)] = 0$, and 
$\mathbb{E}[z_k(t) z_k^*(t+\tau)] = N \delta(\tau)$.
The processes $z_1(t)$, $z_2(t)$ and $z_3(t)$ are statistically independent.

Suppose the transmitted signals are linearly-modulated, i.e.,
the signal sent by transmitter $k$ is
\begin{align}
	A_k(0,t) = \sum_{m=0}^{n-1} x_k[m] \ p(t-m T_s)
\end{align}
where $(x_k[0],\ldots,x_k[n-1])$ is the codeword of transmitter $k$
and $p(t)$ is a pulse such that 
$p(t)=0$ for $t \notin [0,T_s]$ and
\begin{align}
	\int_0^{T_s} |p(\lambda)|^2 d\lambda = E_s.
	\label{eq:nonzerogvm-pulse-energy}
\end{align}

We analyze the setup above in two steps.
\begin{enumerate}
\item 
We start with a simplified version in Sec. \ref{sec:zero-gvm}
where GVM is neglected, i.e., $\beta_{1k}$ is taken to be the same for all $k$ 
so that $d_{kj}=0$ for all $k$,$j$.
For simplicity, we use a rectangular pulse $p(t)$
and consider mainly two WDM channels.
\item
We use the insights gained from Sec. \ref{sec:zero-gvm} to
address the three-user model with GVM and general (time-limited) pulses
in Sec. \ref{sec:nonzero-gvm}.
\end{enumerate}
Table \ref{tab:dt-models} summarizes the assumptions.
\begin{table}
\begin{center}
\begin{tabular}{|lll|l|}
\hline
 GVM 	& Pulse 	& Users 	& Sec.\\
 \hline
 No  	& rectangular 	& two 		& \ref{sec:dt-model-2user}	\\
 No		& rectangular 	& three 	& \ref{sec:dt-model-kuser}	\\
 Yes	& general time-limited 	& three & \ref{sec:dt_model}		\\
 \hline
\end{tabular}	
\end{center}
\caption{Assumptions for Discrete-Time Models}
\label{tab:dt-models}
\end{table}

\section{Zero Group Velocity Mismatch}
\label{sec:zero-gvm}
Consider zero GVM 
with a rectangular pulse (in the time domain) and $E_s=1$.
We present a discrete-time two-user channel model in Sec. \ref{sec:dt-model-2user}, 
and we show that a pre-log 1/2 is achievable for two users by using either
pure amplitude modulation (Sec. \ref{sec:zero-gvm-am}) or
pure phase modulation (Sec. \ref{sec:onering}). 
We introduce interference focusing in Sec. \ref{sec:multiring}
and show that it achieves a pre-log 1 for both users, and therefore no degrees of freedom are lost.
An extension of the discrete-time model to three users is presented in Sec. \ref{sec:dt-model-kuser}.

\subsection{Discrete-Time Two-User Model}
\label{sec:dt-model-2user}
Consider a two-user system in which
receiver $k$, $k=1,2$, obtains $(Y_{k}[0], Y_{k}[1], \cdots, Y_{k}[n-1])$ by matched filtering the received signal $r_k(t)$ and sampling the filter output at the symbol rate.
Equations (\ref{eq:Ak_solution}--\ref{eq:phi2_solution}) and (\ref{eq:rx_k}),
with $\beta_{11} = \beta_{12}$ and $A_3(0,t)=0$,
imply that the channel is memoryless. Hence, we drop the time indices and write the input-output relationships as
\begin{align}
	Y_1 & = X_1 \exp\left({i h_{11} |X_1|^2+ i h_{12} |X_2|^2}\right) + Z_1 \label{eq:channel2a} \\
	Y_2 & = X_2 \exp\left({i h_{21} |X_1|^2+ i h_{22} |X_2|^2}\right) + Z_2 \label{eq:channel2b}
\end{align}
where $Z_{k}$ is circularly-symmetric complex Gaussian noise with variance $N$.
The noise random variables at the receivers are independent.
The term $\exp(i h_{kk} |X_{k}|^2)$ models SPM and the term $\exp(i h_{k\ell} |X_{\ell}|^2)$, $k \neq \ell$, models XPM. We regard the $h_{k\ell}$ as \textit{channel coefficients} that are time invariant. These coefficients are known at the transmitters as well as the receivers.
We use symmetric power constraints
\begin{align} \label{eq:power-constraint-kuser}
	\mathbb{E}\left[ |X_{k}|^2 \right] \le P, \quad k=1,2
\end{align}
but the results below generalize to asymmetric powers.

A \textit{scheme} is a collection $\{ (\mathcal{C}_1(P,N),\mathcal{C}_2(P,N) )\}$ of pairs of codes
indexed by $(P,N)$, such that user $k$ uses the code $\mathcal{C}_k(P,N)$ 
that satisfies the power constraint and achieves an information rate $R_k(P,N)$ where $k=1,2$.
We distinguish between two limiting cases:
1) fixed noise with growing powers and 
2) fixed powers with vanishing noise.

\begin{definition} \label{def:highp-prelog}
The \textit{high-power} pre-log pair $(\overline{r}_1, \overline{r}_2)$ is achieved by a scheme 
if the rates satisfy
\begin{equation}
\overline{r}_k(N) = \lim_{P \rightarrow \infty} \frac{R_k(P,N)}{\log(P/N)}
\text{ for } k=1,2.
\label{eq:high-power_prelog_userk_def}
\end{equation}

\end{definition}
\begin{definition} \label{def:lowp-prelog}
The \textit{low-noise} pre-log pair $(\underline{r}_1, \underline{r}_2)$ is achieved by a scheme 
if the rates satisfy
\begin{equation}
\underline{r}_k(P) = \lim_{N \rightarrow 0} \frac{R_k(P,N)}{\log(P/N)}
\text{ for } k=1,2.
\label{eq:low-noise_prelog_userk_def}
\end{equation}
\end{definition}

The (high-power or low-noise) pre-log pair $(1/2,1/2)$ can be achieved if both users use amplitude modulation only or phase modulation only, 
as shown in Sec. \ref{sec:zero-gvm-am} and Sec. \ref{sec:onering}, respectively. 
We show in Sec. \ref{sec:multiring} that the high-power pre-log pair $(1,1)$ can be achieved 
through \textit{interference focusing}.

\subsection{Amplitude Modulation}
\label{sec:zero-gvm-am}
First, we introduce a result by Lapidoth \cite[Sec. IV]{LapidothPhaseNoise2002}.
\begin{lemma}
Let $Y=X+Z$ where $Z$ is a circularly-symmetric complex Gaussian random variable with mean $0$ and variance $N$.
Define $S \equiv |X|^2/P$. Suppose $S$ is distributed as
\begin{align}
 p_S(s) = \frac{e^{-s/2}}{\sqrt{2 \pi s}}, \quad s \geq 0.
 \label{eq:gamma_dist_1dof}
\end{align}
In other words, $|X|^2$ follows a Gamma distribution (or a Chi-squared distribution) with \emph{one} degree of freedom and has mean $P$.
Then we have
\begin{align}
	I(|X|^2;|Y|^2) \geq \frac{1}{2} \log\left(\frac{P}{2 N} \right) + o(1)
\label{eq:awgn-amplitude-rate-lowerbound}
\end{align}
where $o(1)$ tends to zero as $P/N$ tends to infinity.
\hfill $\blacksquare$
\end{lemma}

If $|X_1|^2/P$ and $|X_2|^2/P$ are distributed according to $p_S$ in (\ref{eq:gamma_dist_1dof}), 
then we have for $k=1,2$
\begin{align}
	I(X_k;Y_k)
	\geq I(|X_k|;|Y_k|)  
  = I(|X_k|^2;|Y_k|^2)
\end{align}
and it follows from (\ref{eq:awgn-amplitude-rate-lowerbound}) that
\begin{align}
	I(X_k;Y_k)
	&\ge \frac{1}{2} \log\left(\frac{P}{2 N} \right) + o(1).
\end{align}
It follows that the high-power and low-noise pre-log pair $(1/2,1/2)$ can be achieved when both users use amplitude modulation.

\subsection{Phase Modulation}
\label{sec:onering}
Suppose the transmitters use phase modulation with $|X_1| = \sqrt{P}$ and $|X_2| = \sqrt{P}$.
The input-output equations (\ref{eq:channel2a})--(\ref{eq:channel2b}) become
\begin{align}
	Y_1 & = X_1 ~ e^{i h_{11} P + i h_{12} P } + Z_1  \label{eq:Y1_ring}\\
	Y_2 & = X_2 ~ e^{i h_{21} P + i h_{22} P } + Z_2. \label{eq:Y2_ring}
\end{align}
Therefore, each receiver sees a constant phase shift which allows us to treat each transmitter-receiver pair separately as an AWGN channel.
We next show that the pre-log pair $(r_1,r_2)=(1/2,1/2)$ can be achieved by using phase modulation only.

\begin{theorem}[One-Ring Modulation]
\label{theorem:one-ring}
Fix $P>0$. Let $Y=X+Z$ where $Z$ is a circularly-symmetric complex Gaussian random variable with mean $0$ and variance $N$,
and $X = \sqrt{P} e^{i \Phi_X}$ where $\Phi_X$ is a real random variable uniformly distributed on $[0,2\pi)$.
Then we have
\begin{align}
	I(X;Y)
	&\ge \frac{1}{2} \log\left(\frac{2 P}{N} \right) - 1 ~ (\text{nats}).
	\label{eq:one_ring_lb}
\end{align}
\end{theorem}
\begin{proof}
We have
\begin{align}
	I(X;Y) 
	&= \mathbb{E}[-\log p_Y(Y)] - \log(\pi e N)	 \label{eq:I_PhiX_Y}		
\end{align}
The pdf $p_Y$ of $Y$ can be shown to be \cite[p. 688]{JLT2010}
\begin{align}
	p_Y(y) 
	&=  \frac{1}{\pi N} e^{-(y_A^2+P)/N} I_0\left(\frac{2 y_A \sqrt{P} }{N}\right)
\end{align}
where $I_0(\cdot)$ is the modified Bessel function of the first kind of order zero
and $Y_A = |Y|$.
Therefore, we have
\begin{align}
	h(Y) 
	&=  \mathbb{E}\left[ 
	-\log\left(\frac{1}{\pi N} e^{-(Y_A^2+P)/N} I_0\left(\frac{2 Y_A \sqrt{P} }{N}\right) \right)
	\right] \nonumber \\
	&\stackrel{(a)}{\geq}
	\mathbb{E}\left[ 
	-\log\left(\frac{1}{\pi N} \frac{e^{-(Y_A - \sqrt{P})^2 /N}}{\sqrt{2 Y_A \sqrt{P} /N}} \right)
	\right] \nonumber \\
	&\stackrel{(b)}{\geq}
	\mathbb{E}\left[  \log\left({\pi N} {\sqrt{2 Y_A \sqrt{P}/N}} \right) \right] \nonumber \\
	& =
	\frac{1}{4} \log\left(\frac{2P}{N}\right)
	+
	\log\left({\pi N}\right)
	+
	\frac{1}{2} \mathbb{E} \left[ \log\left(Y_A \sqrt{\frac{2}{N}}\right) \right]
	\label{eq:entropy_y_lowerbound}
\end{align}
where $(a)$ follows by using Lemma \ref{lemma:besseli_lb} in Appendix A
and ($b$) holds because $(Y_A-\sqrt{P})^2 \geq 0$.
The pdf of $Y_A$ is given by
\begin{align}
  p_{Y_A}(y_A)
	&= \int_{-\pi}^{\pi} p_{Y}(y) y_A d\phi_y \nonumber\\
	&= \frac{2 y_A}{N} e^{-(y_A^2+P)/N} 
	I_0\left(\frac{2 y_A \sqrt{P}}{N}\right) .
	\label{eq:pya}
\end{align}
The last expectation in \eqref{eq:entropy_y_lowerbound} is
\begin{align}
	\mathbb{E} \left[ \log\left(Y_A \sqrt{\frac{2}{N}}\right) \right]
	&= \int_{0}^{\infty} p_{Y_A}(y_A)
	\log\left(y_A \sqrt{\frac{2}{N}}\right) \, dy_A
	\nonumber \\
	&\stackrel{(a)}{=} \int_{0}^{\infty} z e^{-(z^2 + \nu^2)/2}
	I_0\left( z \nu \right) \log(z) \, dz \nonumber \\
	&\stackrel{(b)}{=} \frac{1}{2} \left[ \Gamma\left(0,\frac{P}{N}\right) + \log\left(\frac{2P}{N}\right) \right] 
	\nonumber \\
	&\stackrel{(c)}{\geq} \frac{1}{2} \log\left(\frac{2P}{N}\right)
	\label{eq:elogya_lowerbound}
\end{align}
where $\Gamma(a,x)$ is the upper incomplete Gamma function, see (\ref{eq:upper-incomplete-gamma}) below.
Step ($a$) follows by setting $\nu^2=2P/N$ and $z = y_A \sqrt{2/N}$,
($b$) follows from Lemma \ref{lemma:integral-lemma3} in Appendix B and
($c$) holds because $\Gamma\left(0,x\right) \ge 0$ for $x\ge0$.\footnote{Note that $\lim_{x \rightarrow \infty} \Gamma\left(0,x\right) = 0$.} 
Combining (\ref{eq:I_PhiX_Y}), (\ref{eq:entropy_y_lowerbound}) and (\ref{eq:elogya_lowerbound}) 
concludes the proof.
\end{proof}

Now, suppose that $X_k = \sqrt{P} e^{i \Phi_{X,k}}$ for $k=1,2$ where 
$\Phi_{X,1}$ and $\Phi_{X,2}$ are statistically independent and uniformly distributed on $[0,2\pi)$.
It follows from (\ref{eq:Y1_ring}), (\ref{eq:Y2_ring}) and Theorem \ref{theorem:one-ring} that
the high-power and low-noise pre-log 
pair $(1/2,1/2)$ can be achieved  when both users use phase modulation.

\subsection{Interference Focusing}
\label{sec:multiring}
We propose an \textit{interference focusing} technique in which the transmitters \textit{focus} their phase interference on one point by constraining their transmitted signals to satisfy
\begin{align}
	h_{21} |X_1|^2 & = 2 \pi \tilde{n}_1, ~ \tilde{n}_1=1,2,3,\ldots \label{eq:ring-constraint-1} \\
	h_{12} |X_2|^2 & = 2 \pi \tilde{n}_2, ~ \tilde{n}_2=1,2,3,\ldots \label{eq:ring-constraint-2} 
\end{align}
In other words, the transmitters use \textit{multi-ring modulation} with specified spacings between the rings.\footnote{Multi-ring modulation was used in~\cite{JLT2010,Essiambre2008PRL,Essiambre2009} for symmetry and computational reasons. We here find that it is useful for improving rate.} We thereby remove XPM interference and \eqref{eq:channel2a}-\eqref{eq:channel2b} reduce to
\begin{align}
	Y_k = X_k e^{i h_{kk} |X_k|^2} + Z_k, \quad k=1,2.
\end{align}
This channel is effectively an AWGN channel since $h_{kk}$ is known by receiver $k$ and the SPM phase shift is determined by the desired signal $X_k$.
We will show that the high-power pre-log pair $(1,1)$ is achieved under the constraints \eqref{eq:ring-constraint-1}-\eqref{eq:ring-constraint-2}.

\begin{theorem}[Multi-Ring Modulation]
\label{theorem:multi-ring}
Let $Y=X+Z$ where $Z$ is a circularly-symmetric complex Gaussian random variable with mean $0$ and variance $N$.
Suppose $\mathbb{E}[|X|^2] \leq P$ and $|X|^2$ is allowed to take on values that are multiples of a fixed real number $\hat{p}>0$,
i.e., $|X|^2 = m \hat{p}$ where $m \in \mathbb{N}$. Then there exists a probability distribution $p_X$ of $X$ such that
\begin{align}
	\lim_{P \rightarrow \infty} \frac{I(X;Y)}{\log(\snr)}
	&\ge 1.
	\label{eq:multi_ring_lb}
\end{align}
\end{theorem}
\begin{proof}
Define $X_A=|X|$ and $\Phi_X = \arg{X}$.
Consider multi-ring modulation, i.e.,
$X_A$ and $\Phi_X$ are statistically independent,
$\Phi_X$ is uniformly distributed on the interval $[0,2\pi)$
and $X_A \in \{\sqrt{\Pu_j}:j=1,\ldots,J\}$
where $J$ is the number of rings.
We choose the rings to be spaced uniformly in amplitude as
\begin{align} \label{eq:power-levels}
	\Pu_j = a j^2 \, \hat{p}
\end{align}
where $a$ is a positive integer. 
We further use a uniform frequency of occupation of rings with $P_{X_A}(\sqrt{\Pu_j}) = 1/J$, $j=1,2,\ldots,J$. 
The power constraint is therefore
\begin{align} \label{eq:pwr_constraint_uniform}
	\frac{1}{J} \sum_{j=1}^J a j^2 \, \hat{p} \le P .
\end{align}
For (\ref{eq:pwr_constraint_uniform}), we compute
\begin{align}
	\frac{1}{J} \sum_{j=1}^{J} a \hat{p} \, j^2 = a \hat{p} \frac{(J+1)(2J+1)}{6}
\end{align}
and to satisfy the power constraint we choose\footnote{The solution for $J$ should be positive and rounded down to the nearest integer but we ignore these issues for notational simplicity.}
\begin{align}
	J = \frac{-3  + \sqrt{1 + 48 P/(a \hat{p})}}{4}.
\end{align}
Moreover, we choose $a = \lfloor \max\{1,N\log(\snr)\} \rfloor$.
We remark that we say $f(x)$ scales as $g(x)$ if
\begin{align*}
\lim_{x \rightarrow \infty} \frac{f(x)}{g(x)} = \text{constant}.
\end{align*}
For example, $J$ scales as $\sqrt{(\snr)/\log(\snr)}$ when $a$ is chosen as above, i.e., we have
\begin{align*}
\lim_{P \rightarrow \infty} \frac{J}{\sqrt{(\snr)/\log(\snr)}} = \text{constant}.
\end{align*}

We have
\begin{align}
	I(X;Y) 
	&= I(X_A  \Phi_X;Y) \nonumber \\
	& = I(X_A;Y) + I(\Phi_X;Y|X_A) \label{eq:mult-ring-rate}
\end{align}
The term $I(X_A ; Y)$ can be viewed as the amplitude contribution while the term $I(\Phi_X ; Y | X_A)$ is the phase contribution.

\subsubsection{Phase Contribution}
We show that the phase modulation contributes at least 1/2 to the pre-log
when using multi-ring modulation.

\begin{lemma} \label{lemma:sum-bounds-integral}
For integers $a$ and $b$ with $a \leq b$, a non-decreasing function $f(x)$ in $x$ satisfies
\begin{align}
	\int_{a-1}^b f(x) dx \leq \sum_{i=a}^b f(i).
\end{align}
\hfill $\blacksquare$
\end{lemma}
We thus have
\begin{align}
	I(\Phi_X;Y|X_A) 
	& \stackrel{(a)}{=} 	 \sum_{j=1}^J \frac{1}{J} \, I(\Phi_X;Y|X_A=\sqrt{\Pu_j}) \nonumber \\
	& \stackrel{(b)}{\geq} \sum_{j=1}^J \frac{1}{J}  \, \frac{1}{2} \log\left( \frac{\Pu_j}{N} \right) - 1 \nonumber \\
	& \stackrel{(c)}{=} 	 \frac{1}{2 J} \sum_{j=1}^J \log\left(\frac{a j^2 \hat{p}}{N}\right) - 1 \nonumber \\
	& \stackrel{(d)}{\geq} \frac{1}{2 J} \int_{0}^J \log\left(\frac{a x^2 \hat{p}}{N}\right) dx - 1 \nonumber \\
	& \stackrel{(e)}{=}		 \frac{1}{2} \log\left( \frac{a J^2 \hat{p}}{N e^2} \right) - 1 
\end{align}
where
($a$) follows from the uniform occupation of rings,
($b$) follows from Theorem \ref{theorem:one-ring},
($c$) holds by choosing the rings according to \eqref{eq:power-levels},
($d$) follows from Lemma \ref{lemma:sum-bounds-integral} since the logarithm is an increasing function and
($e$) follows by using $\log(a x^2 \hat{p}/N) = \log(a \hat{p}/N) + 2 \log(x)$ and
\begin{align}
\int \log(x) dx = x \log\left({x/e}\right).
\end{align}
We can therefore write
\begin{align}
	\lim_{P \rightarrow \infty}  \frac{I(\Phi_X;Y|X_A)}{\log(\snr)} 
	&\ge \lim_{P \rightarrow \infty} \frac{\frac{1}{2} \log(a J^2 \hat{p}/N)}{\log(\snr)} 
	= \frac{1}{2} \label{eq:phase-prelog}
\end{align}
where \eqref{eq:phase-prelog} follows because $a$ scales as $N\log(\snr)$, $J^2$ scales as $(\snr)/\log(\snr)$,  and $\hat{p}$ is independent of $P$ and $N$.
The pre-log of the phase contribution is therefore at least $1/2$.

\subsubsection{Amplitude Contribution}
We show that amplitude modulation contributes $1/2$ to the pre-log.
We have
\begin{align}
	I(X_A;Y)  = H(X_A) - H(X_A|Y)
\end{align}
where $H(X_A) = \log(J)$. We showed previously that $J$ scales as $\sqrt{(\snr)/\log(\snr)}$ if $a$ scales as $N\log(\snr)$. We bound $H(X_A|Y)$ using Fano's inequality as
\begin{align}
	H(X_A|Y) &\le H(X_A|\hat{X}_A) \nonumber \\
	&\le H(P_e) + P_e \log(J-1)
	\label{eq:fano}	
\end{align}
where $\hat{X}_A$ is any estimate of $X_A$ given $Y$, $P_e=\Pr[\hat{X}_A \ne X_A]$ and $H(P_e)$ is the binary entropy function with a general logarithm base. Suppose we use the minimum distance estimator
\begin{align}
	\hat{X}_A = \arg \min_{x_A \in \mathcal{X}_A} |Y_A-x_A|
	\label{eq:min_distance_dec}
\end{align}
where $Y_A=|Y|$ and $\mathcal{X}_A = \{\sqrt{\Pu_j}:j=1,\ldots,J\}$. 
The probability of error $P_e$ is upper bounded by 
(see Lemma \ref{lemma:min-distance-estimator} in Appendix \ref{sec:appendix-min-distance})
\begin{align}
	P_e \leq \frac{2}{J} \sum_{j=2}^{J} \exp\left(-\frac{\Delta_j^2}{4}\right)
\end{align}
where $\Delta_j = (\sqrt{\Pu_j} - \sqrt{\Pu_{j-1}})/\sqrt{N}$.
For the power levels (\ref{eq:power-levels}), we have $\Delta_j = \sqrt{a \hat{p}/N}$ for all $j$, and hence
\begin{align} \label{eq:Pe-bound-2}
	P_e 
	\le \frac{2(J-1)}{J} \exp\left(-\frac{a \hat{p}}{4 N}\right) 
	\le       2          \exp\left(-\frac{a \hat{p}}{4 N}\right) .
\end{align}
We see from \eqref{eq:Pe-bound-2} that $\lim_{P \rightarrow \infty} P_e = 0$ if $a$ scales as $N\log(\snr)$
(recall that $J$ scales as $\sqrt{(\snr)/\log(\snr)}\,$). 
We thus have $\lim_{P \rightarrow \infty} H(X_A|Y) = 0$ by using (\ref{eq:fano}).
Consequently, we have
\begin{align} \label{eq:amplitude-prelog}
	\lim_{P \rightarrow \infty} \frac{I(X_A;Y)}{\log(\snr)} 
	= \lim_{P \rightarrow \infty} \frac{\log(J)}{\log(\snr)} 
	= \frac{1}{2}.
\end{align}
Finally, combining \eqref{eq:mult-ring-rate}, \eqref{eq:phase-prelog}, and \eqref{eq:amplitude-prelog} gives (\ref{eq:multi_ring_lb}).
\end{proof}

We conclude that interference focusing achieves the largest-possible high-power pre-log of 1. Each user can therefore exploit all the phase and amplitude degrees of freedom simultaneously.

\subsection{Discrete-Time Three-User Model}
\label{sec:dt-model-kuser}
Consider a WDM system with three users.
Receiver $k$ obtains $(Y_{k}[0], Y_{k}[1], \cdots, Y_{k}[n-1])$ by matched filtering the received signal $r_k(t)$ in (\ref{eq:rx_k}) and sampling the filter output at the symbol rate.
By setting $\beta_{11} = \beta_{12} = \beta_{13}$
in (\ref{eq:Ak_solution}--\ref{eq:phi3_solution}), 
we have the following memoryless channel model:
\begin{align}
	Y_{k} = X_{k} \exp\left( i \sum_{\ell=1}^3 h_{k\ell} |X_{\ell}|^2 \right) + Z_{k} 
	\label{eq:channel}
\end{align}
for $k=1,2,3$ where $Z_{k}$ is circularly-symmetric complex Gaussian noise with variance $N$.
All noise random variables at different receivers are statistically independent. The terms $\exp(i h_{kk} |X_{k}[j]|^2)$ model SPM and the terms $\exp(i h_{k\ell} |X_{\ell}[j]|^2)$, $\ell \neq k$, model XPM. The $h_{k\ell}$ are again \textit{channel coefficients} that are time invariant and are known at the transmitters as well as the receivers.
The power constraints are
\begin{align} \label{eq:power-constraint}
	\mathbb{E}\left[ |X_{k}|^2 \right] \le P, \quad k=1,2,3.
\end{align}

\subsubsection*{Interference Focusing}
We outline how to apply interference focusing to the three-user channel. Define the interference phase vector
\begin{align}
   \underline{\Psi} \stackrel{\Delta}{=} [\Psi_1,\Psi_2,\Psi_3]^T
\end{align}
where $\Psi_k = \sum_{\ell=1}^3 h_{k\ell} |X_{\ell}|^2$ and the instantaneous power vector
\begin{align}
  \underline{\Pi} \stackrel{\Delta}{=} \left[ |X_1|^2,|X_2|^2,|X_3|^2 \right]^T.
\end{align}
The relationship between the $\underline{\Psi}$ and $\underline{\Pi}$ in matrix form is
\begin{align}
	\underline{\Psi} = H_{SP} \, \underline{\Pi}  + H_{XP} \, \underline{\Pi} 
\end{align}
where $H_{SP}$ is a diagonal matrix that accounts for SPM and $H_{XP}$ is a zero-diagonal matrix that accounts for XPM.
For example, suppose the XPM matrix for a 3-user interference network is
\begin{align}
	H_{XP} = \left[
	\begin{array}{ccc}
	0       &      1/2     &      3/5     \\
	3/4     &      0       &      2/3     \\
	5/6     &      1/5     &      0   	
	\end{array}
	\right].
\end{align}
Suppose that each transmitter knows the channel coefficients between itself and all the receiving nodes. The transmitters can thus use power levels of the form
\begin{align}
   \underline{\Pi} & = 2 \pi \cdot \left[ \, \text{lcm}(4,6) m_1, \text{lcm}(2,5) m_2, \text{lcm}(5,3) m_3 \,\right] \nonumber \\
   & = 2 \pi \cdot \left[\,  12 m_1, 10 m_2, 15 m_3 \,\right] 
\end{align}
where $\text{lcm}(a,b)$ is the least common multiple of $a$ and $b$, and $m_1,m_2,m_3$ are positive integers. We thus have
\begin{align}
	H_{XP} \, \underline{\Pi} 
	=
	2 \pi
	\left[
	\begin{array}{ccc}
	0      &      5     &      9     \\
	9   	 &      0     &     10     \\
	10     &      2     &      0   	
	\end{array}
	\right]
	\left[
	\begin{array}{c}
	m_1	\\
	m_2	\\
	m_3     
	\end{array}
	\right]
\end{align}
which implies that the phase interference has been eliminated.

The above example combined with an analysis similar to Section \ref{sec:multiring} shows that interference focusing will give each user a pre-log of $1$ even for three-user interference networks. However, the XPM coefficients $h_{k\ell}$ must be \textit{rationals}. This result can be generalized to the $K$-user case. Modifying interference focusing for \textit{real-valued} XPM coefficients is an interesting problem. It is clear from the example that interference focusing does not require global channel state information.

\section{Non-Zero Group Velocity Mismatch}
\label{sec:nonzero-gvm}
We next consider non-zero GVM, i.e., $\beta_{13} \neq \beta_{12} \neq \beta_{11}$. 
Without loss of generality, suppose that $\beta_{13} > \beta_{12} > \beta_{11}$.
We now use a general time-limited pulse $p(t)$.

We start with the continuous-time model in Sec. \ref{sec:ct_model} below and derive a discrete-time model in Sec. \ref{sec:dt_model}.
We show that a pre-log 1/2 is achievable for all users by using pure amplitude modulation in Sec. \ref{sec:inner-bound}.
Next, we show that interference focusing achieves a pre-log of at least 1 for all users under certain conditions in Sec. \ref{sec:ifocus}.
Finally, we show in Sec. \ref{sec:outer-bound} that interference focusing achieves the maximum pre-log of 1 and, 
therefore, interference focusing is pre-log optimal.

\subsection{Continuous-Time Model}
\label{sec:ct_model}
The signal $r_k(t)$ in (\ref{eq:rx_k}) is fed to a bank of linear time-invariant (LTI) filters with impulse responses 
$\{h_{\idx}(\T)\}_{\idx \in \mathcal{F}_k}$, 
where $\mathcal{F}_k \subset \mathbb{Z}=\{\ldots,-1,0,1,\ldots\}$
and
\begin{align}
h_{\idx}(t) = p^*(-t) \exp(-i 2\pi {\idx} \K(-t))
\end{align}
where $K(t)$ is defined as
\begin{align}
	\K(t) = \frac{1}{E_s} \int_{0}^{t} |p(\lambda)|^2 d\lambda.
	\label{eq:rise_func}
\end{align}
The choice of the set $\mathcal{F}_k$ is specified in Sec. \ref{sec:ifocus}.
We show in Appendix \ref{sec:ortho_impulses} that the impulse responses of the filters are orthogonal, i.e., if $\idx_1 \neq \idx_2$, then we have
\begin{align}
	\int_{-\infty}^{\infty} h_{\idx_1}(\xi) h^*_{\idx_2}(\xi) d\xi
	= 0.
	\label{eq:ortho_diff}
\end{align}

The remaining analysis is similar for all receivers, hence we present the analysis for receiver 1 only.
The output of the filter with index ${\idx}$ is
\begin{align}
	y_{1,{\idx}}(\T) = r_1(\T) \star h_{\idx}(\T)
\end{align}
where $\star$ denotes convolution.
The noiseless part $\tilde{y}_{1,{\idx}}(\T)$ of the output of this filter is
\begin{align}
	\ifdefined\twocolumnmode & \fi
	\tilde{y}_{1,{\idx}}(\T+\beta_{11} L) 
	\ifdefined\twocolumnmode \nonumber\\ \fi 
	&\stackrel{\Delta}{=} 
		 A_1(L,\T+\beta_{11} L) \star h_{\idx}(\T)	\nonumber\\
	&= \left( A_1(0,\T) e^{i \phi_1(L,\T)} \right) \star h_{\idx}(\T)  \nonumber \\
	&= \left( \sum_{m=0}^{n-1} x_1[m] \ p(\T-m T_s) e^{i \phi_1(L,\T)} \right) \star h_{\idx}(\T) \nonumber \\
	&= \sum_{m=0}^{n-1} x_1[m] 
	\int p(\tau-m T_s) p^*(\tau-\T) 
	e^{i \phi_1(L,\tau) - i 2\pi \idx \K(\tau-\T)} d\tau
\end{align}
where the integral is over the whole real line.
Sampling the output signal $y_{1,{\idx}}(\T+\beta_{11} L)$ at the time instants $\T = j T_s$, for $j=1,2,\ldots,n$, yields
\begin{align}
	\ifdefined\twocolumnmode & \fi
	\tilde{y}_{1,{\idx}}(j T_s+\beta_{11} L) 
	\ifdefined\twocolumnmode \nonumber\\ \fi
	&= x_1[j] \ 
	\int_{j T_s}^{j T_s+T_s} 
	|p(\tau-j T_s)|^2 e^{i \phi_1(L,\tau) - i 2\pi \idx \K(\tau-j T_s)} d\tau
	\label{eq:y1nj_noiseless_integral}
\end{align}
where we used $p(t)=0$ for $t \notin [0,T_s]$.
We write $\phi_1(L,\tau)$ as
\begin{align}     
\phi_1(L,\tau) = \phi_{11}(L,\tau) + \phi_{12}(L,\tau) + \phi_{13}(L,\tau)
\label{eq:phi1_ct}
\end{align}
where we have defined 
\begin{align}
 \phi_{11}(L,t) &\stackrel{\Delta}{=}   \gamma_1  L \ |A_1(0,t)|^2
 \label{eq:phi11_definition} \\
 \phi_{12}(L,t) &\stackrel{\Delta}{=} 2 \gamma_1 L_{12} \
 \frac{1}{T_s} \int_{t-L d_{21}}^{t} |A_2(0,\lambda)|^2 d\lambda
 \label{eq:phi12_definition} \\
 \phi_{13}(L,t) &\stackrel{\Delta}{=} 2 \gamma_1 L_{13} \
 \frac{1}{T_s} \int_{t-L d_{31}}^{t} |A_3(0,\lambda)|^2 d\lambda
 \label{eq:phi13_definition} 
\end{align}
and where $L_{1k} \stackrel{\Delta}{=} T_s/|d_{1k}|$ for $k \neq 1$. 
Since $p(t)=0$ for $t \notin [0,T_s]$, we have
\begin{align}
 \phi_{12}(L,t) 
 &= \frac{2 \gamma_1 L_{12}}{T_s} \int_{t-L d_{21}}^{t} \sum_{m=0}^{n-1} |x_2[m]|^2 \ |p(\lambda-m T_s)|^2 d\lambda 	\nonumber\\
 &= \frac{2 \gamma_1 L_{12}}{T_s} \sum_{m=0}^{n-1} 
 			|x_2[m]|^2 \ \int_{t-L d_{21}}^{t} |p(\lambda-m T_s)|^2 d\lambda	\nonumber\\
 &= 2 \gamma_1 L_{12} \frac{E_s}{T_s} \sum_{m=0}^{n-1} |x_2[m]|^2 \ \psi(t-m T_s;d_{21})
 \label{eq:psi12_expansion}
\end{align}
where $\psi(t;d)$ is defined as
\begin{align}
 \psi(t;d) \stackrel{\Delta}{=} \frac{1}{E_s} \int_{t-L d}^{t} |p(\lambda)|^2 d\lambda.
 \label{eq:psi12m_definition}
\end{align}
If $L d \geq T_s$, then
\begin{align}
  \psi(t;d) =
  \left\{
  \begin{array}{ll}
   \K(t), 					& 0 \leq t < T_s	\\
   1, 							& T_s \leq t < L d	\\
	 \tilde{\K}(t;d),	& L d \leq t < L d + T_s \\
   0, 							& \text{otherwise}
  \end{array}
  \right.
\label{eq:psi12m_Ld_gt_Ts}
\end{align}
where $\K(t)$ is defined by (\ref{eq:rise_func}) and $\tilde{\K}(t;d)$ is given by
\begin{align}
	\tilde{\K}(t;d) 
	&= \frac{1}{E_s} \int_{t-L d}^{T_s} |p(\lambda)|^2 d\lambda	\nonumber \\
	&= \frac{1}{E_s} \int_{0}^{T_s} |p(\lambda)|^2 d\lambda - \frac{1}{E_s} \int_{0}^{t-L d} |p(\lambda)|^2 d\lambda \nonumber \\
	&= 1 - \K(t-L d).
	\label{eq:fall_func}
\end{align}
One can express $\phi_{13}(L,t)$ in a similar manner.
Suppose that $L |d_{1k}| = M_{1k} T_s$ for some positive integer $M_{1k}$ for $k=2,3$.
Hence, for $ \tau \in [j T_s, j T_s+T_s]$, 
we have\footnote{We use the convention of setting the quantities that involve a negative time index to zero.}
\ifdefined\twocolumnmode{
\begin{align}
 \phi_{11}(L,\tau)
  &=   \gamma_1 L \frac{E_s}{T_s} |x_1[j]|^2 		\label{eq:phi11}\\
 \phi_{12}(L,\tau) 
  &=   2 \gamma_1 L_{12} \frac{E_s}{T_s} \Big( 
    \Big( \sum_{r=1}^{M_{12}} |x_2[j-r]|^2 \Big) + 	\nonumber \\& \qquad
    \left( |x_2[j]|^2 - |x_2[j-M_{12}]|^2 \right) \K(t-j T_s) 
		\Big) \label{eq:phi12}\\
 \phi_{13}(L,\tau) 
  &=   2 \gamma_1 L_{13} \frac{E_s}{T_s} \Big( 
		\Big( \sum_{r=1}^{M_{13}} |x_3[j-r]|^2 \Big) + 	\nonumber \\& \qquad
    \left( |x_3[j]|^2 - |x_3[j-M_{13}]|^2 \right) \K(t-j T_s)  
  \Big) .
	\label{eq:phi13}
\end{align}
}\else{
\begin{align}
 \phi_{11}(L,\tau)
  &=   \gamma_1 L \frac{E_s}{T_s} |x_1[j]|^2 		\label{eq:phi11}\\
 \phi_{12}(L,\tau) 
  &=   2 \gamma_1 L_{12} \frac{E_s}{T_s} \Big( 
    \Big( \sum_{r=1}^{M_{12}} |x_2[j-r]|^2 \Big) + 	
    \left( |x_2[j]|^2 - |x_2[j-M_{12}]|^2 \right) \K(t-j T_s) 
		\Big) \label{eq:phi12}\\
 \phi_{13}(L,\tau) 
  &=   2 \gamma_1 L_{13} \frac{E_s}{T_s} \Big( 
		\Big( \sum_{r=1}^{M_{13}} |x_3[j-r]|^2 \Big) + 	
    \left( |x_3[j]|^2 - |x_3[j-M_{13}]|^2 \right) \K(t-j T_s)  
  \Big) .
	\label{eq:phi13}
\end{align}
}\fi
By substituting (\ref{eq:phi11})--(\ref{eq:phi13}) in (\ref{eq:phi1_ct}), we get
\begin{align}
	\phi_1(L,\tau) = \phi_1[j] + 2 \pi v_1[j] \K(\tau-j T_s)
\end{align}
where
\ifdefined\twocolumnmode{
\begin{align} 
\phi_1[j]
&= h_{11} |x_1[j]|^2
 + h_{12} \sum_{r=1}^{M_{12}} |x_2[j-r]|^2 
 \nonumber\\ &\qquad\qquad
 + h_{13} \sum_{r=1}^{M_{13}} |x_3[j-r]|^2 \label{eq:phi1j} \\  
v_1[j]
&= h_{12} \left( |x_2[j]|^2 - |x_2[j-M_{12}]|^2 \right)/2\pi
 \nonumber\\&
 + h_{13} \left( |x_3[j]|^2 - |x_3[j-M_{13}]|^2 \right)/2\pi \label{eq:v1j}
\end{align}
and
\begin{align}
 h_{11} &= \gamma_1 L \frac{E_s}{T_s}, \nonumber \\
 h_{12} &= 2 \gamma_1 L_{12} \frac{E_s}{T_s}, \nonumber \\
 h_{13} &= 2 \gamma_1 L_{13} \frac{E_s}{T_s}.
\label{eq:hcoeffs}
\end{align}
}\else{
\begin{align} 
\phi_1[j]
&= h_{11} |x_1[j]|^2
 + h_{12} \sum_{r=1}^{M_{12}} |x_2[j-r]|^2 
 + h_{13} \sum_{r=1}^{M_{13}} |x_3[j-r]|^2 \label{eq:phi1j} \\  
v_1[j]
&= h_{12} \left( |x_2[j]|^2 - |x_2[j-M_{12}]|^2 \right)/2\pi
 + h_{13} \left( |x_3[j]|^2 - |x_3[j-M_{13}]|^2 \right)/2\pi \label{eq:v1j}\\
 h_{11} &= \gamma_1 L \frac{E_s}{T_s} , \quad 
 h_{12} = 2 \gamma_1 L_{12} \frac{E_s}{T_s}, \quad 
 h_{13} = 2 \gamma_1 L_{13} \frac{E_s}{T_s}.
\label{eq:hcoeffs}
\end{align}
}\fi
Then by substituting in (\ref{eq:y1nj_noiseless_integral}), we have
\begin{align}
	\ifdefined\twocolumnmode & \fi
	\tilde{y}_{1,{\idx}}(j T_s+\beta_{11} L) 
	\ifdefined\twocolumnmode \nonumber\\ \fi
	&= x_1[j] \ E_s \ e^{i \phi_1[j]} \ \int_{0}^{T_s} \frac{|p(\tau)|^2}{E_s} e^{i 2\pi(v_1[j]-\idx) \K(\tau)} d\tau.
	\label{eq:y1n_integral}
\end{align}
By applying Lemma \ref{lemma:leibniz} in Appendix \ref{sec:ortho_impulses}
to evaluate the integral in (\ref{eq:y1n_integral}),
the noiseless part $\tilde{y}_{1,{\idx}}[j]$ of the output of the filter with index ${\idx}$ at time $j$ 
can be written as
\begin{align}
&\tilde{y}_{1,{\idx}}[j]   = 
x_1[j] E_s e^{i\phi_1[j]} \ u_{1,{\idx}}[j]
\label{eq:y1nj_without_noise}
\end{align}
where
\begin{align}
  u_{1,{\idx}}[j] 
  & = \left\{
  \begin{array}{ll}
	 \displaystyle
        \frac{\exp\left(i 2\pi(v_1[j] - {\idx})\right) - 1 }
	      {          			i 2\pi(v_1[j] - {\idx})            }	,&\text{ if } v_1[j] \neq {\idx} \\
  			1																											,&\text{ otherwise}.
 \end{array}
 \right. 
 \label{eq:u1nj}
\end{align}
The output of the filter with index ${\idx}$ at time $j$ is
\begin{align}
	y_{1,{\idx}}[j] = y_{1,{\idx}}(j T_s+\beta_{11} L) = \tilde{y}_{1,{\idx}}[j] + z_{1,{\idx}}[j]
	\label{eq:y1nj}
\end{align}
where 
\begin{align}
	z_{1,{\idx}}[j] = \left. z_1(\T) \star h_{\idx}(\T) \right|_{\T=j T_s+\beta_{11} L}.
	\label{eq:z1nj}
\end{align}
The variable $z_{1,{\idx}}[j]$ is Gaussian with mean 0 and variance $N E_s$.
Moreover, due to the orthogonality of the filter bank impulse responses, we have
$\mathbb{E}\big[ z_{1,{\idx_1}}[j] z_{1,\idx_2}^*[j] \big] = 0$ for all ${\idx_1} \neq \idx_2$,
which implies that 
the random variables $\{z_{1,{\idx}}[j]\}_{\idx \in \mathcal{F}_1}$ are independent.

\subsection{Discrete-Time Model}
\label{sec:dt_model}
The input $x_k[j]$ of transmitter $k$ to the channel at time $j$ is a scalar,
whereas the channel output $\yv{k}[j]$ at receiver $k$ at time $j$ is a vector
whose components are $y_{k,{\idx}}[j]$, $f \in \mathcal{F}_k$.
To compute mutual information, we now consider the codeword $X_k^{\blk} = (X_{k}[1], X_{k}[2], \cdots, X_{k}[{\blk}])$
and the receiver samples $\Yv{k}^{\blk} = (\Yv{k}[1], \Yv{k}[2], \cdots, \Yv{k}[{\blk}])$ as random variables.
The input-output relations are
\begin{align}
&Y_{k,{\idx}}[j] = X_k[j] ~ e^{ i \Phi_k[j] } ~ U_{k,{\idx}}[j] + Z_{k,{\idx}}[j]
\label{eq:model_yknj}
\end{align}
with
\ifdefined\twocolumnmode{
\begin{align}
\Phi_1[j]
&= h_{11} |X_1[j]|^2
+ h_{12} \sum_{r=1}^{M_{12}} |X_2[j-r]|^2
\nonumber\\&\
+ h_{13} \sum_{r=1}^{M_{13}} |X_3[j-r]|^2 
\\
\Phi_2[j]
&= h_{21} \sum_{r=1}^{M_{12}} |X_1[j+M_{12}-r]|^2
+ h_{22} |X_2[j]|^2
\nonumber\\&\
+ h_{23} \sum_{r=1}^{M_{23}} |X_3[j-r]|^2
\\
\Phi_3[j]
&= h_{31} \sum_{r=1}^{M_{13}} |X_1[j+M_{13}-r]|^2
\nonumber\\&\
+ h_{32} \sum_{r=1}^{M_{23}} |X_2[j+M_{23}-r]|^2
+ h_{33} |X_3[j]|^2
\label{eq:Phi_kj}
\end{align}
}\else{
\begin{align}
\Phi_1[j]
&= h_{11} |X_1[j]|^2
+ h_{12} \sum_{r=1}^{M_{12}} |X_2[j-r]|^2
+ h_{13} \sum_{r=1}^{M_{13}} |X_3[j-r]|^2 
\\
\Phi_2[j]
&= h_{21} \sum_{r=1}^{M_{12}} |X_1[j+M_{12}-r]|^2
+ h_{22} |X_2[j]|^2
+ h_{23} \sum_{r=1}^{M_{23}} |X_3[j-r]|^2
\\
\Phi_3[j]
&= h_{31} \sum_{r=1}^{M_{13}} |X_1[j+M_{13}-r]|^2
+ h_{32} \sum_{r=1}^{M_{23}} |X_2[j+M_{23}-r]|^2
+ h_{33} |X_3[j]|^2
\label{eq:Phi_kj}
\end{align}
}\fi
where $M_{12}$, $M_{13}$ and $M_{23}$ are positive integers and
\begin{align}
  U_{k,{\idx}}[j] 
  & = \left\{
  \begin{array}{ll}
	 \displaystyle
        \frac{\exp\left(i 2\pi(V_k[j] - {\idx})\right) - 1 }
	      {          			i 2\pi(V_k[j] - {\idx})            }	,&\text{ if } V_k[j] \neq {\idx} \\
  			1																											,&\text{ otherwise}
 \end{array}
 \right. 
 \label{eq:model_uknj}
\end{align}
where we define 
\ifdefined\twocolumnmode{
\begin{align}
	V_1[j] &\stackrel{\Delta}{=} h_{12}(|X_2[j]				|^2 - |X_2[j-M_{12}]|^2)/ 2\pi
														\nonumber\\&
														+  h_{13}(|X_3[j]				|^2 - |X_3[j-M_{13}]|^2)/ 2\pi \\
	V_2[j] &\stackrel{\Delta}{=} h_{21}(|X_1[j+M_{12}]|^2 - |X_1[j]				|^2)/ 2\pi
														\nonumber\\&
														+  h_{23}(|X_3[j]				|^2 - |X_3[j-M_{23}]|^2)/ 2\pi \\
	V_3[j] &\stackrel{\Delta}{=} h_{31}(|X_1[j+M_{13}]|^2 - |X_1[j]				|^2)/ 2\pi
														\nonumber\\&
														+  h_{32}(|X_2[j+M_{23}]|^2 - |X_2[j]				|^2)/ 2\pi.														
	\label{eq:model_vk}
\end{align}
}\else{
\begin{align}
	V_1[j] &\stackrel{\Delta}{=} h_{12}(|X_2[j]				|^2 - |X_2[j-M_{12}]|^2)/ 2\pi
														+  h_{13}(|X_3[j]				|^2 - |X_3[j-M_{13}]|^2)/ 2\pi \\
	V_2[j] &\stackrel{\Delta}{=} h_{21}(|X_1[j+M_{12}]|^2 - |X_1[j]				|^2)/ 2\pi
														+  h_{23}(|X_3[j]				|^2 - |X_3[j-M_{23}]|^2)/ 2\pi \\
	V_3[j] &\stackrel{\Delta}{=} h_{31}(|X_1[j+M_{13}]|^2 - |X_1[j]				|^2)/ 2\pi
														+  h_{32}(|X_2[j+M_{23}]|^2 - |X_2[j]				|^2)/ 2\pi.														
	\label{eq:model_vk}
\end{align}
}\fi
$Z_{k,{\idx}}[j]$ models the noise at filter $\idx$ of receiver $k$ at time $j$, 
and the random variables $\{Z_{k,{\idx}}[j]\}_{k,\idx,j}$ are independent circularly-symmetric complex Gaussian random variables with mean 0  and variance $N$.
We regard the $h_{k\ell}$ as \textit{channel coefficients} that are time invariant and known globally.
The following symmetric power constraints are imposed:
\begin{align} \label{eq:power-constraint-gvm}
	\frac{1}{{\blk}} \sum_{j=1}^{{\blk}} \mathbb{E}\left[ |X_{k}[j]|^2 \right] \le P, \quad k=1,2,3.
\end{align}

A \textit{scheme} is a collection 
$\{ (\mathcal{C}_1(P,N),\mathcal{C}_2(P,N),\mathcal{C}_3(P,N) )\}$ of triples of codes
indexed by $(P,N)$, such that user $k$ uses the code $\mathcal{C}_k(P,N)$ 
that satisfies the power constraint and achieves an information rate $R_k(P,N)$ for $k=1,2,3$
where
\begin{align}
R_k(P,N) = I(X_k;\Yv{k}) \equiv \lim_{n \rightarrow \infty} \frac{1}{n} I(X_k^n;\Yv{k}^n).
\end{align}

We extend the definitions of pre-logs made in Definitions \ref{def:highp-prelog} and \ref{def:lowp-prelog}.
\begin{definition}
The \textit{high-power} pre-log triple $(\overline{r}_1, \overline{r}_2, \overline{r}_3)$ is achieved by a scheme 
if the rates satisfy
\begin{equation}
\overline{r}_k(N) = \lim_{P \rightarrow \infty} \frac{R_k(P,N)}{\log(P/N)}
\text{ for } k=1,2,3.
\end{equation}
\end{definition}
\begin{definition}
The \textit{low-noise} pre-log triple $(\underline{r}_1, \underline{r}_2, \underline{r}_3)$ is achieved by a scheme 
if the rates satisfy
\begin{equation}
\underline{r}_k(P) = \lim_{N \rightarrow 0} \frac{R_k(P,N)}{\log(P/N)}
\text{ for } k=1,2,3.
\end{equation}
\end{definition}

The (high-power or low-noise) pre-log triple $(1/2,1/2,1/2)$ can be achieved if all users use phase modulation only (see Sec. \ref{sec:inner-bound}).
It is not obvious whether $(1/2,1/2,1/2)$ is achievable by using amplitude modulation only, e.g., such as in Sec. \ref{sec:zero-gvm-am}. 
This is because $U_{k,{\idx}}[j]$ in (\ref{eq:model_yknj}) has a random amplitude.
We show in Sec. \ref{sec:ifocus} that the high-power pre-log triple $(1,1,1)$ can be achieved 
for any positive $N$ through \textit{interference focusing}.

\subsection{Inner Bound: Phase Modulation}
\label{sec:inner-bound}
Suppose we use only the filter with index $f=0$.
Suppose further that the inputs $X_k^n$ of user $k$ are i.i.d.
with a constant amplitude $\sqrt{P}$ and a uniformly random phase (a ring), i.e., we have
\begin{align}
X_k[j] = \sqrt{P} e^{i \Phi_{X,k}[j]}
\end{align}
where $\Phi_{X,k}[j]$ is uniform on $[-\pi,\pi)$ for $j=1,2,\ldots,n$.
Therefore, the outputs become
\begin{align}
&Y_{k,0}[j] = X_k[j] ~ e^{ i \Phi_k[j] } ~ U_{k,0}[j] + Z_{k,0}[j]
\end{align}
with
\begin{align}
\Phi_1[j] &= [ h_{11} + h_{12} M_{12} + h_{13} M_{13} ] P \nonumber
\\
\Phi_2[j] &= [ h_{21} M_{12} + h_{22} + h_{23} M_{23} ] P \nonumber
\\
\Phi_3[j] &= [ h_{31} M_{13} + h_{32} M_{23} + h_{33} ] P
\end{align}
i.e., the phase $\Phi_k[j]$ is constant for all $j=1,\ldots,n$.
Moreover, we have
\begin{align}
U_{1,0}[j] &= 1, \qquad \max\{M_{12},M_{13}\} < j \leq n	\nonumber \\
U_{2,0}[j] &= 1, \qquad M_{23} < j < n-M_{12} 				\nonumber \\
U_{3,0}[j] &= 1, \qquad 1 \leq j < n-\max\{M_{13},M_{23}\}.
\end{align}
Thus, the users are decoupled under constant amplitude modulation, except near the beginning and the end of transmission.
We have
\begin{align}
\frac{1}{n} I(X_1^n;\Yv{1}^n)
&\stackrel{(a)}{\geq}	\frac{1}{n} I(X_1^n;Y_{1,0}^n) \nonumber \\
&\stackrel{(b)}{\geq} \frac{1}{n} \sum_{j=1}^{n} I(X_1[j];Y_{1,0}[j]) \nonumber \\
&\stackrel{(c)}{\geq} \frac{1}{n} \sum_{j=\max\{M_{12},M_{13}\}+1}^{n} I(X_1[j];Y_{1,0}[j]) \nonumber \\
&\stackrel{(d)}{\geq} \left( \frac{n-\max\{M_{12},M_{13}\}}{n} \right) \left[ \frac{1}{2} \log\left(\frac{P}{N}\right) - 1\right]
\end{align}
where 
$(a)$ follows from the chain rule and the non-negativity of mutual information,
$(b)$ follows because $X_1,\ldots,X_n$ are i.i.d. and because conditioning does not increase entropy,
$(c)$ follows from the non-negativity of mutual information and
$(d)$ holds because (see Theorem \ref{theorem:one-ring})
\begin{align}
I(X_1[j];Y_{1,0}[j]) \geq \frac{1}{2} \log\left(\frac{2 P}{N}\right) - 1.
\end{align}
As $n \rightarrow \infty$, we have
\begin{align}
R_1(P,N)
&\geq \frac{1}{2} \log\left(\frac{2 P}{N}\right) - 1.
\end{align}
By using similar steps for users 2 and 3, we have
\begin{align}
R_k(P,N)
&\geq \frac{1}{2} \log\left(\frac{2 P}{N}\right) - 1
\end{align}
for $k=1,2,3$ which implies that the pre-log triple $(1/2,1/2,1/2)$ is achieved by using one receiver filter and phase modulation.

\subsection{Interference Focusing}
\label{sec:ifocus}
We use \textit{interference focusing}, i.e., we focus the phase interference on \textit{one} point 
by imposing the following constraints on the transmitted symbols:
\begin{align}
 h_{21} |X_1[j]|^2 = 2\pi \tilde{n}_{21},\ h_{31} |X_1[j]|^2 = 2\pi \tilde{n}_{31},	\nonumber \\
 h_{12} |X_2[j]|^2 = 2\pi \tilde{n}_{12},\ h_{32} |X_2[j]|^2 = 2\pi \tilde{n}_{32},	\nonumber \\
 h_{13} |X_3[j]|^2 = 2\pi \tilde{n}_{13},\ h_{23} |X_3[j]|^2 = 2\pi \tilde{n}_{23},
\end{align}
where
$\tilde{n}_{21}$, $\tilde{n}_{31}$,
$\tilde{n}_{12}$, $\tilde{n}_{32}$,
$\tilde{n}_{13}$ and $\tilde{n}_{23} \in \mathbb{N}$,
which ensures that the XPM interference is eliminated.
Suppose that
$h_{21}$, $h_{31}$,
$h_{12}$, $h_{32}$,
$h_{13}$ and $h_{23}$
are rational.
Then the interference focusing constraints become
\begin{align}
 |X_k[j]|^2 = 2\pi \hat{p}_k ~ \tilde{n}_{k}
\end{align}
where
\begin{align}
 \hat{p}_1 \stackrel{\Delta}{=} \text{lcm}(\text{den}(h_{21}),\text{den}(h_{31})), \label{eq:p1_hat}\\
 \hat{p}_2 \stackrel{\Delta}{=} \text{lcm}(\text{den}(h_{12}),\text{den}(h_{32})), \label{eq:p2_hat}\\
 \hat{p}_3 \stackrel{\Delta}{=} \text{lcm}(\text{den}(h_{13}),\text{den}(h_{23}))  \label{eq:p3_hat}
\end{align}
where $\text{den}(x)$ is the denominator of a rational number $x$.

Because of the power constraint, 
only a subset $\mathcal{P}_k$ of the allowed rings is actually used.
In this case, $V_k[j] \in \mathcal{V}_k$, for $k=1,2,3$, where 
\begin{align}
	\mathcal{V}_k = \left\{ \sum_{j \neq k} d_j: d_j \in \mathcal{D}_j, j \in\{1,2,3\} \right\}
\end{align}
and
\begin{align}
	\mathcal{D}_k = \left\{\frac{p-p^\prime}{\hat{p_k}}: p \in \mathcal{P}_k, p^\prime \in \mathcal{P}_k \right\}
\end{align}
which leads us to choose the sets of ``normalized frequencies'' $\mathcal{F}_k$
of the filter banks at the receivers as
$\mathcal{F}_k = \mathcal{V}_k$.

Thus, under interference focusing, the output at receiver $k$ at time $j$ is a vector $\Yv{k}[j]$, 
whose components are $\{Y_{k,{\idx}}[j]\}_{{\idx} \in \mathcal{V}_k}$,
where
\begin{align}
Y_{k,{\idx}}[j]   = X_k[j] \ \exp\left( i h_{kk} |X_k[j]|^2 \right) U_{k,{\idx}}[j] + Z_{k,{\idx}}[j]
\label{eq:ykn}
\end{align}
and where 
\begin{align}
  U_{k,{\idx}}[j]
  = \left\{ \begin{array}{ll}
		1	,&\text{ if } V_k[j] = {\idx}, \\
  	0	,&\text{ otherwise.}
 \end{array} \right.
 \label{eq:u_kn_ifocus}
\end{align}
This means that
exactly one filter (the filter with index $V_k[j]$) output among all the filters contains the signal corrupted by noise,
while all other filters put out noise.
Therefore, we have
\begin{align}
	&\frac{1}{n} I(X_k^n;\Yv{k}^n) \nonumber\\
	&\stackrel{(a)}{\geq} 
	\frac{1}{n} \sum_{j=1}^n I(X_k[j];\Yv{k}[j]) \nonumber\\
	&\stackrel{(b)}{\geq}
	\frac{1}{n} \sum_{j=1}^n  I\left(X_k[j]; Y_{k,V_k[j]}[j] \right) \nonumber\\
	&\stackrel{(c)}{=} I\left( X_k[1];X_k[1] e^{i h_{kk}|X_k[1]|^2}+Z_{k,V_k[1]}[1] \right)
\end{align}
where
$(a)$ follows because the $X_k^n$ are i.i.d. and because conditioning does not increase entropy;
$(b)$ follows from the chain rule and the non-negativity of mutual information
(it can be shown that equality holds, see Appendix \ref{sec:sufficient})
; and
$(c)$ holds because the $X_k^n$ are i.i.d. and the channel becomes a memoryless time-invariant channel under interference focusing.
It follows from Theorem \ref{theorem:multi-ring} that by using interference focusing,
we have
\begin{align}
	\lim_{P \rightarrow \infty} \frac{I\left( X_k[1];X_k[1] e^{i h_{kk}|X_k[1]|^2}+Z_{k,V_k[1]}[1] \right)}{\log(P/N)}
	\geq 1
\end{align}
which implies that $\overline{r}_k \geq 1$ for $k=1,2,3$.
Hence, the high-power pre-log triple $(1,1,1)$ is achievable.
We again remark that the above analysis generalizes for different power constraints at the transmitters.
However, the question of whether all users can simultaneously achieve a low-noise pre-log of 1 is open for both models with and without GVM.

The following (downsized) example illustrates our receiver structure, 
and the role that interference focusing plays in choosing its parameters.

\textit{Example:}
Consider 2 transmitters that use a rectangular pulse, i.e., suppose
\begin{align}
 p(t) = \left\{ 
  \begin{array}{ll}
  \displaystyle \sqrt{E_s/T_s} ,	& 0 \leq t < T_s \\
  0,					& \text{otherwise}
  \end{array}
 \right.
\label{eq:square_pulse_def}
\end{align}
where the power constraints are $P_1 = 8$ and $P_2 = 7$
on transmitter 1 and 2, respectively.
Suppose that
$h_{12} = 5$, 
$h_{21} = 4$.
Since this is a two-user system, we may use 
(\ref{eq:ring-constraint-1}) and (\ref{eq:ring-constraint-2}) rather than (\ref{eq:p1_hat}) and (\ref{eq:p2_hat}), i.e.,
we use $\hat{p}_1 = 1/h_{21} = 0.25$ and $\hat{p}_2 = 1/h_{12} = 0.2$.
Suppose that the users choose 
the power levels 
$\mathcal{P}_1 = \{2\pi\hat{p}_{1} \tilde{n}_1: \tilde{n}_1 = 1,4,9\} = \{ 0.5\pi, 2\pi, 4.5\pi \}$
and
$\mathcal{P}_2 = \{2\pi\hat{p}_{2} \tilde{n}_2: \tilde{n}_2 = 2,8\} = \{ 0.8\pi, 3.2\pi \}$ 
(see Fig. \ref{fig:constell}).
These choices satisfy the power constraints 
and eliminate the interference.
The parameters of the filter banks are 
$\mathcal{F}_1 = \mathcal{V}_1 = \{-6,0,6\}$ and
$\mathcal{F}_2 = \mathcal{V}_2 = \{-8,-5,-3,0,3,5,8\}$.
In other words, receiver 1 has 3 filters
whose frequency responses are sinc functions
centered at $f_1-6/T_s$, $f_1$, and $f_1+6/T_s$,
whereas
receiver 2 has 7 filters
whose frequency responses are sinc functions
centered at 7 different frequencies 
(see Fig. \ref{fig:filter_bank}).
This shows that, because of the nonlinearity, 
the receivers need to extract information from a ``bandwidth'' larger than
the ``bandwidth'' of the transmitted signal.

\begin{figure}
	\centering
		\includegraphics[width= {\ifdefined\twocolumnmode 0.49\columnwidth \else 0.36\textwidth \fi} ]{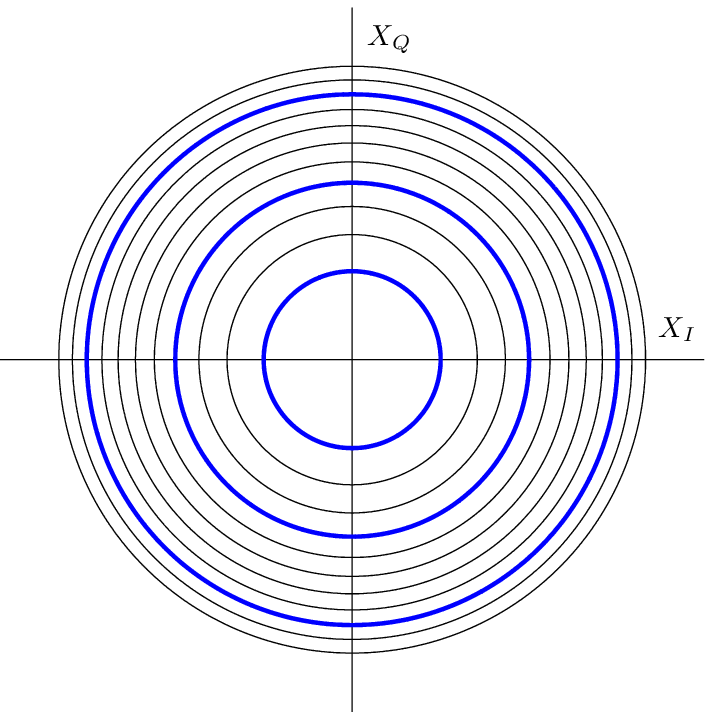}
		\includegraphics[width= {\ifdefined\twocolumnmode 0.49\columnwidth \else 0.36\textwidth \fi} ]{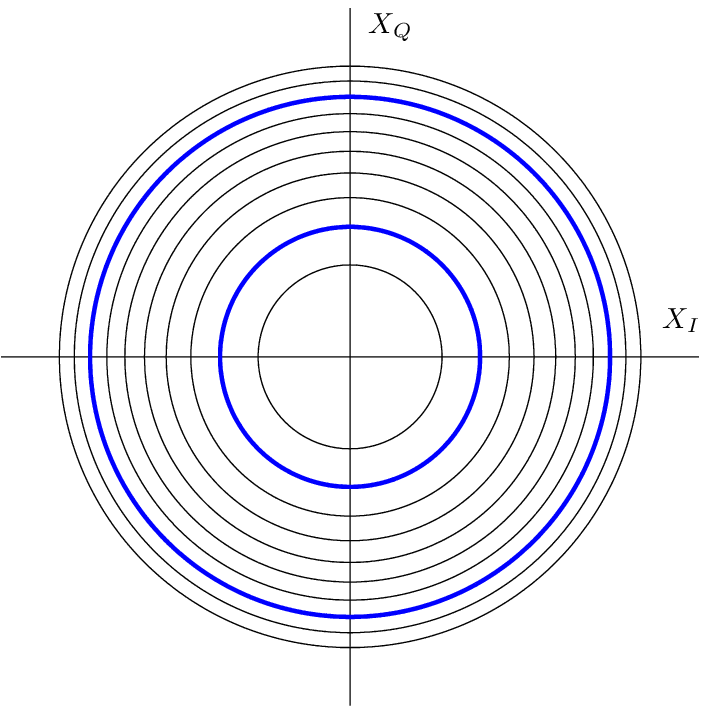}
	\caption{Ring modulation used by transmitter 1 (left) and transmitter 2 (right).
	The thin lines are the rings allowed by interference focusing, 
	and the thick blue lines are the rings selected for transmission.}
	\label{fig:constell}
\end{figure}

\begin{figure}
	\centering
		\includegraphics[width= {\ifdefined\twocolumnmode 1.0\columnwidth \else 0.75\textwidth \fi} ]{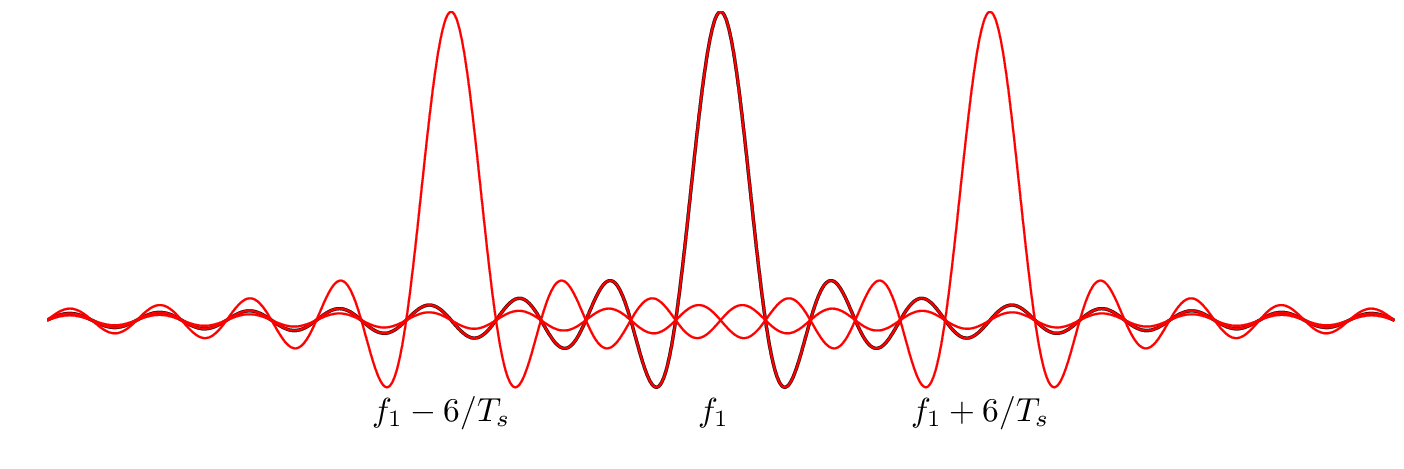} \\
		\includegraphics[width= {\ifdefined\twocolumnmode 1.0\columnwidth \else 0.75\textwidth \fi} ]{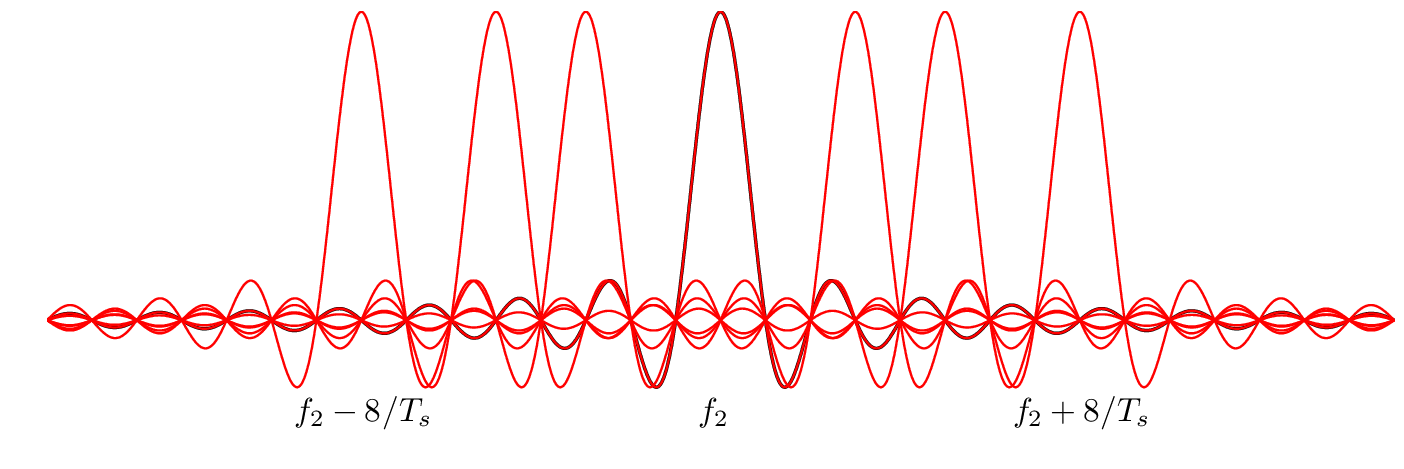}
	\caption{Frequency responses of the filters at receivers 1 (top) and 2 (bottom).}
	\label{fig:filter_bank}
\end{figure}

\subsection{Outer Bound}
\label{sec:outer-bound}
\subsubsection{Interference Focusing}
We show next that the maximal pre-log triple for the model of Sec. \ref{sec:dt_model} is $(1,1,1)$ 
when interference focusing is used.
We have
\begin{align}
	 I( X_1^n ; \Yv{1}^n )
	&\stackrel{(a)}{\leq} 
	I( X_1^n ; \Yv{1}^n , V_1^n) \nonumber\\
	&\stackrel{(b)}{=}  
	I( X_1^n ; \Yv{1}^n | V_1^n) \nonumber\\
	&\stackrel{(c)}{=}    
	I( X_1^n ; Y_{1,V_1[1]} Y_{1,V_1[2]} \ldots Y_{1,V_1[n]} | V_1^n) \nonumber\\	
	&\stackrel{(d)}{=}    
	I( X_1^n ; Y_{1,V_1[1]} Y_{1,V_1[2]} \ldots Y_{1,V_1[n]} ) \nonumber\\		
	&{\leq} n \log\left( 1 + \frac{P}{N} \right).
\end{align}
Step $(a)$ follows from the chain rule and the non-negativity of mutual information;
$(b)$ holds because $X_1^n$ is independent of $V_1^n$;
$(c)$ holds because $Y_{1,f}[j] = Z_{1,f}[j]$ for $f \neq V_1[j]$
and the variables $\{ Z_{1,f}[j] : j=1,\ldots,n \text{ and } f \neq V_1[j] \}$ are independent of
$Y_{1,V_1[1]}, Y_{1,V_1[2]}, \ldots, Y_{1,V_1[n]}$ and $X_1^n$ and
$(d)$ holds because $X_1^n$ and $Y_{1,V_1[1]}$, $Y_{1,V_1[2]}$, $\ldots$, $Y_{1,V_1[n]}$ are independent of $V_1^n$ 
(which follows from $Z_{1,V_1[1]}, Z_{1,V_1[2]}, \ldots, Z_{1,V_1[n]}$ being i.i.d.).
By using a similar argument for receiver 2 and receiver 3, we eventually have
\begin{align}
	R_k \leq \log\left( 1 + \frac{P}{N} \right)
\end{align}
for $k=1,2,3$ which implies that the maximal pre-log triple is $(1,1,1)$.

\subsubsection{General Modulation}
We show next that the maximal pre-log triple is $(1,1,1)$ 
for any modulation scheme.
We use a genie-aided strategy. Suppose a genie reveals the codewords $x_2^n$ and $x_3^n$ of users 2 and 3 to receiver 1 prior to transmission.
Receiver 1 generates $\phi_{12}(t)$ and $\phi_{13}(t)$ according to (\ref{eq:phi12_definition}) and (\ref{eq:phi13_definition}), respectively,
and uses them to cancel XPM in the received signal, i.e., receiver 1 generates
\begin{align}
\tilde{r}_1(t) = r_1(t) ~ e^{-j \phi_{12}(t) -j \phi_{13}(t)}.
\end{align}
The XPM-free signal $\tilde{r}(t)$ is fed to a filter with an impulse response $p^*(-t)$ and the output of the filter is sampled at symbol rate.
Matched filtering with symbol rate sampling does not incur any information loss because XPM is canceled.
The $j$-th filter output is 
\begin{align}
\tilde{y}_1[j] 
&= \left. \tilde{r}_1(t) \star p^*(-t) \right|_{t=j \Ts} \nonumber \\
&= x_1[j] e^{j h_{11} |x_1[j]|^2} + \tilde{z}_1[j]
\label{eq:ytilde_genie}
\end{align}
where $\tilde{z}_1[j]$ is a realization of a Gaussian random variable with mean zero and variance $N$.
The channel (\ref{eq:ytilde_genie}) is a memoryless AWGN channel and therefore we have
\begin{align}
	 I( X_1^n ; \tilde{Y}_1[1] \tilde{Y}_1[2] \ldots \tilde{Y}_1[n])
	&{\leq} n \log\left( 1 + \frac{P}{N} \right).
\end{align}
Similarly, it can be shown that the maximum pre-log for users 2 and 3 is 1,
implying that the maximal pre-log triple is $(1,1,1)$.

\section{Discussion}
\label{sec:ialign}

Interference focusing is reminiscent of interference alignment,
which refers to techniques of signal construction so that
undesired signals at each receiver arrive along the same dimensions
while the desired signal can be resolved through the remaining dimensions.
We highlight the main differences between interference alignment and interference focusing.

Cadambe and Jafar \cite{CJ08} introduced an (asymptotic) interference alignment scheme 
for $K$-user single-input single-output linear interference channels 
which achieves the optimal degrees of freedom (DoF)
as long as the channel coefficients are time-varying (or frequency-selective).
Their scheme relies on beamforming over symbol extensions 
to separate signal spaces based on linear independence.

For constant channels, especially with real channel coefficients, 
interference alignment along linearly independent dimensions may not achieve the optimal DoF.
Motahari et al \cite{MotahariRealIA} developed interference alignment along rationally independent dimensions 
to achieve the optimal DoF.
Their approach is referred to as \emph{real} interference alignment.

Interference alignment lets each user achieve half the DoF that can be achieved in the absence of all other interferers
(colloquially: each user gets half ``the cake'').
In contrast, orthogonalization techniques, e.g., time division or frequency division, split the resources among $K$ users
so that each user gets only 1/$K$ of the resources.

In optical fiber, splitting the bandwidth among $K$ users by using WDM 
does not guarantee that each user gets 1/$K$ of the interference-free capacity (IFC)
because of the fiber nonlinearity. 
In the zero-dispersion case, for example, each user gets 1/(2$K$) of the IFC using Gaussian modulation.
We have shown interference focusing enables each user to get 1/$K$ of the IFC, but not half of it, at high SNR.
We remark that interference focusing requires neither symbol extensions nor global CSI.

\section{Conclusion}
\label{sec:conclusion}

We introduced two discrete-time interference channel models
based on a simplified optical fiber model. 
We used coupled differential equations derived from the NLS equation to develop our models.
In the first model, there was no dispersion.
This discrete-time model was justified 
by using a rectangular pulse shape at the transmitters
and matched filters at the receivers.
The nonlinear nature of the fiber-optic medium causes the users to suffer from amplitude-dependent phase interference. 
We introduced a new technique called interference focusing that lets the users take advantage of all the available amplitude and phase degrees of freedom
at high transmission powers. 
In the second model, the second-order dispersion is negligible.
However, we included non-zero GVM
as well as nonlinearity. 
We justified this discrete-time model 
by using a time-limited pulse shape at the transmitters
and a bank of ``frequency-shifted'' matched filters at each receiver.
We proved that all users can achieve a high-power pre-log of 1 simultaneously
by using interference focusing.
We also showed that interference focusing is optimal (for the model of Sec. \ref{sec:dt_model}) in the pre-log sense.

\appendices

\section{Upper Bound on the Modified Bessel Function of the First Kind of Order Zero}
\begin{lemma} \label{lemma:besseli_lb}
\begin{align}
	I_0(z) \leq \frac{e^z}{\sqrt{z}}, \quad z > 0
\end{align}
\end{lemma}
\begin{proof}
We have
\begin{align}
&\cos{\theta} \leq 1 - 4 \theta^2/\pi^2, & 0 \leq \theta \leq \pi/2 		\label{eq:cos_ub_domain1}\\
&\cos{\theta} \leq 0                   , & \pi/2 \leq \theta \leq \pi		\label{eq:cos_ub_domain2}
\end{align}
where (\ref{eq:cos_ub_domain1})
follows by the infinite product form
\begin{align}
	\cos x = \prod_{n=1}^\infty \left[ 1- \frac{4x^2}{\pi^2 (2 n -1)^2}\right].
\end{align}
We thus have
\begin{align}
	I_0(z)
	& = \frac{1}{\pi} \int_{0}^{\pi} e^{z \cos\theta} \, d\theta \nonumber \\
	& \stackrel{(a)}{\le} \frac{1}{\pi} \left[ \int_{0}^{\pi/2} e^{z (1- 4 \theta^2/\pi^2)} \, d\theta + \int_{\pi/2}^{\pi} \, d\theta \right] \nonumber \\
	& =									  \frac{\sqrt{\pi}}{4} \frac{e^z}{\sqrt{z}} \left( 1 - 2 Q(\sqrt{2z}) \right) + \frac{1}{2} \nonumber \\
	& \stackrel{(b)}{\le} \frac{\sqrt{\pi}}{2} \frac{e^z}{\sqrt{z}} \nonumber \\
	&\leq \frac{e^z}{\sqrt{z}}
\end{align}
where
\begin{align}
Q(z) = \int_{z}^{\infty} \frac{1}{\sqrt{2 \pi}} e^{-x^2/2}  dx.
\end{align}
Step ($a$) follows because the exponential function is a monotonic increasing function and 
by (\ref{eq:cos_ub_domain1})--(\ref{eq:cos_ub_domain2}),
while
($b$) holds because $Q(z) \ge 0$ and
\begin{align}
\frac{\sqrt{\pi}}{4} \frac{e^z}{\sqrt{z}}
\geq \frac{\sqrt{\pi}}{4} \min_{z \geq 0} \frac{e^z}{\sqrt{z}}
= \frac{\sqrt{2 e \pi}}{4} 
\ge \frac{1}{2}.
\end{align}
\end{proof}

\section{Expected Value of the Logarithm of a Rician R.V.}

Consider the following functions.
\begin{itemize}
\item 
Gamma function $\Gamma(z)$ \cite[6.1.1]{Abramowitz1972}
\begin{align}
	\Gamma(z) = \int_{0}^{\infty} t^{z-1} e^{-t} dt
	\label{eq:gamma-func-def}
\end{align}

\item 
Psi (Digamma) function $\psi(z)$ \cite[6.3.1]{Abramowitz1972}
\begin{align}
	\psi(z) =\frac{d}{dz} \ln\Gamma(z) = \frac{ \Gamma^\prime(z) }{\Gamma(z)}  
\end{align}

\item 
Upper incomplete Gamma function $\Gamma(a,x)$ \cite[6.5.3]{Abramowitz1972}
\begin{align}
	\Gamma(a,x) = \int_{x}^{\infty} t^{a-1} e^{-t} dt, \qquad a>0
	\label{eq:upper-incomplete-gamma}
\end{align}

\end{itemize}

We derive several useful lemmas concerning these functions.

\begin{lemma} \label{lemma:integral-lemma1}
	\begin{align}
	  \int_{0}^\infty {e^{- \frac{x^2}{2} } x^{2k+1} \ln(x) dx}
		= 2^{k-1} \left(  \Gamma(k+1) \ln(2) + \Gamma^\prime(k+1) \right) 
	\end{align}	
\end{lemma}
\begin{proof}
Consider
	\begin{align}
	  \mathcal{I}_k
	  &\stackrel{\Delta}{=} \int_{0}^\infty {e^{- \frac{x^2}{2} } x^{2k+1} \ln(x) dx} \nonumber \\
		&\stackrel{(a)}{=} \int_{0}^\infty {e^{- u }(\sqrt{2u})^{2k+1} \ln(\sqrt{2u}) \frac{1}{\sqrt{2u} } du} \nonumber \\
		&= \int_{0}^\infty {e^{- u }(2u)^{k} \frac{1}{2} \left(\ln(2) + \ln(u)\right)  du} \nonumber \\
		&= 2^{k-1} \ln(2) \int_{0}^\infty {e^{-u} u^{k} du} 
		 + 2^{k-1} \int_{0}^\infty {e^{- u }u^{k} \ln(u) du} \nonumber \\
		&\stackrel{(b)}{=} 2^{k-1} \left(  \Gamma(k+1) \ln(2) + \Gamma^\prime(k+1) \right) 
	\end{align}
	where ($a$) follows from the transformation of variables $u = x^2/2$
	and $(b)$ follows by (\ref{eq:gamma-func-def}) and \cite[4.352 (4)]{Gradshtein2007}
	\begin{align}
		\int_{u=0}^\infty {e^{- u }u^{k} \ln(u) du} = \Gamma^\prime(k+1).
	\end{align}
\end{proof}

\begin{lemma} \label{lemma:integral-lemma2}
	\begin{align}
		\sum_{k=0}^\infty \frac{t^k}{k!} \psi(k+1) = e^t ( \Gamma(0,t)+\ln(t) )
	\end{align}
\end{lemma}
\begin{proof}	
	We use the following formula \cite[6.2.1 (60)]{Brychkov2008} 
	\begin{align}
		\sum_{k=1}^\infty \frac{t^k}{k!} \psi(k+a) 
		&= \frac{t}{a} e^t \bigg[ a \psi(a) \frac{1-e^{-t}}{t} \nonumber\\&\qquad
															+ ~ _2F_2(1,1;a+1,2;-t) \bigg]
		\label{eq:psi_series_a}
	\end{align}
	where $_2F_2(a_1,a_2;b_1,b_2;x)$ is the generalized hypergeometric function defined as
	\cite[9.14 (1)]{Gradshtein2007},\cite[p. 674]{Brychkov2008}
	\begin{align}
		_2F_2(a_1,a_2;b_1,b_2;x) = \sum_{k=0}^\infty \frac{(a_1)_k (a_2)_k}{(b_1)_k (b_2)_k} \frac{x^k}{k!}
		\label{eq:hypergeometric_2F2}
	\end{align}
	where $(f)_k \stackrel{\Delta}{=} f(f+1)\ldots (f+k-1)$.	
	Setting $a=1$ in (\ref{eq:psi_series_a}) gives
	\begin{align}
		\sum_{k=1}^\infty \frac{t^k}{k!} \psi(k+1) 
		& = t e^t \left[ \psi(1) \frac{1-e^{-t}}{t} + ~ _2F_2(1,1;2,2;-t) \right] \nonumber \\
		& = e^t \left[ F(t) + (1-e^{-t}) \psi(1) \right]
		\label{eq:psi_series_1}
	\end{align}
	where we defined $F(t)$ as
	\begin{align}
		F(t) 
		&\stackrel{\Delta}{=} t \cdot {_2F_2}(1,1;2,2;-t) \nonumber \\
		&= t ~ \sum_{k=0}^\infty \frac{1}{(k+1)^2} \frac{(-t)^k}{k!} \nonumber \\
		&= \sum_{m=1}^\infty \frac{-1}{m} \frac{(-t)^{m}}{m!}.
		\label{eq:Ft_def}
	\end{align}
	From the Fundamental Theorem of Calculus, we have
	\begin{align}
		\int_{0}^{z} F^\prime(t) ~ dt =	F(z) - F(0)
		\label{eq:theorem_of_calculus}
	\end{align}
	where
	\begin{align}
		F^\prime(t) 
		\stackrel{\Delta}{=} \frac{dF(t)}{dt}
		= \sum_{m=1}^\infty \frac{(-t)^{m-1}}{m!}
		= \frac{1-e^{-t}}{t}
	\end{align}
	and therefore the left-hand side of (\ref{eq:theorem_of_calculus}) is 
	\begin{align}
		\int_{0}^{z} \frac{1-e^{-t}}{t} dt
		&=
		  \int_{z}^{\infty} \frac{e^{-t}}{t} dt
		+ \int_{1}^{z} \frac{1}{t} dt		
		\nonumber\\& \qquad
		+
		\left(
		  \int_{0}^{1} \frac{1-e^{-t}}{t} dt
		- \int_{1}^{\infty} \frac{e^{-t}}{t} dt
		\right) \nonumber \\
		&= \Gamma(0,z) + \ln(z) + \gamma
	\end{align}
	where we used the definition of the upper incomplete Gamma function
	and the following integral form for Euler's constant 
	\cite[8.367 (12)]{Gradshtein2007}
	\begin{align}
		\gamma = \int_0^1 \frac{1-e^{-t}}{t} dt - \int_1^\infty \frac{e^{-t}}{t} dt.
	\end{align}
	Since $F(0)=0$ and $\psi(1) = -\gamma$ \cite[8.367 (1)]{Gradshtein2007}, we have
	\begin{align}
		F(z) = \Gamma(0,z) + \ln(z) - \psi(1).
		\label{eq:Fz_close_form}
	\end{align}
	The lemma follows from (\ref{eq:psi_series_1}) and (\ref{eq:Fz_close_form}).
\end{proof}

\begin{lemma} \label{lemma:integral-lemma3}
	\begin{align}
	\int_{0}^\infty {x e^{- \frac{x^2+\nu^2}{2} } I_0(x\nu) \ln(x) dx} 
	= \frac{1}{2} \left( \Gamma\left( 0,\frac{\nu^2}{2} \right) + \ln(\nu^2) \right)
	\end{align}
\end{lemma}
\begin{proof} 
We compute
\begin{align}
	&   \int_{0}^\infty {x e^{- \frac{x^2+\nu^2}{2} } I_0(x\nu) \ln(x) dx} \nonumber \\
	& \stackrel{(a)}{=} 
		  \int_{0}^\infty {x e^{- \frac{x^2+\nu^2}{2} } \left( \sum_{k=0}^{\infty} \frac{(x^2 \nu^2/4)^k}{(k!)^2} \right) \ln(x) dx}	\nonumber\\
	& = \sum_{k=0}^{\infty} \frac{1}{4^k (k!)^2} 
			\left( \int_{0}^\infty {x e^{- \frac{x^2+\nu^2}{2} }(x \nu)^{2k} \ln(x) dx} \right) \nonumber\\
	& \stackrel{(b)}{=} 
	    \sum_{k=0}^{\infty} \frac{1}{4^k (k!)^2} 
			2^{k-1} e^{- \frac{\nu^2}{2} } \nu^{2k} ( \Gamma(k+1) \ln(2) + \Gamma^\prime(k+1) )  \nonumber\\
	& \stackrel{(c)}{=} 
			e^{ -\frac{\nu^2}{2} } \left[
			\frac{\ln(2)}{2} \sum_{k=0}^{\infty} \frac{(\nu^2/2)^k}{k!}
		+ \frac{1     }{2} \sum_{k=0}^{\infty} \frac{(\nu^2/2)^k}{k!} \psi(k+1)
			\right]			\nonumber \\			
	& \stackrel{(d)}{=} 
			e^{ -\frac{\nu^2}{2} } \left[ 
			\frac{\ln(2)}{2} e^{ \frac{\nu^2}{2} }
		+ \frac{1     }{2} e^{ \frac{\nu^2}{2} } \left( \Gamma\left( 0,\frac{\nu^2}{2} \right) + \ln\left( \frac{\nu^2}{2} \right) \right)
			\right] \nonumber \\			
	& = \frac{1}{2} \left( \Gamma\left( 0,\frac{\nu^2}{2} \right) + \ln(\nu^2) \right)
\end{align}
where in ($a$) we used the series representation of $I_0(\cdot)$
\cite[9.6.10]{Abramowitz1972}
\begin{align}
	I_0(x) = \sum_{k=0}^{\infty} \frac{(x^2/4)^k}{(k!)^2}.
\end{align}
Step
($b$) follows from Lemma \ref{lemma:integral-lemma1},
($c$) follows because \cite[6.1.6]{Abramowitz1972}
\begin{align}
\Gamma(k+1) = k!
\end{align}
and $\psi(k+1) = \Gamma^\prime(k+1)/\Gamma(k+1)$ and
($d$) follows by Lemma \ref{lemma:integral-lemma2}.
\end{proof}

\section{Minimum-Distance Estimator}
\label{sec:appendix-min-distance}

Let $Y=X+Z$ where $Z$ is a circularly-symmetric complex Gaussian random variable with mean $0$ and variance $N$.
Suppose $X_A \stackrel{\Delta}{=} |X| \in \mathcal{X}_A = \{\sqrt{\Pu_j}:j=1,\ldots,J\}$
where $0 < \Pu_1 < \Pu_2 <\ldots <\Pu_J$.
Define the minimum-distance estimator $\hat{X}_A$ as
\begin{align}
	\hat{X}_A = \arg \min_{x_A \in \mathcal{X}_A} |Y_A-x_A|
\end{align}
where $Y_A=|Y|$.
\begin{lemma} \label{lemma:min-distance-estimator}
The probability of error for uniformly distributed $X_A$ satisfies
\begin{align}
	P_e
	\stackrel{\Delta}{=} \Pr[\hat{X}_A \ne X_A]
	\leq \frac{2}{J} \sum_{j=2}^{J} \exp\left(-\frac{\Delta_j^2}{4}\right)
\end{align}
where $\Delta_j = (\sqrt{\Pu_j} - \sqrt{\Pu_{j-1}})/\sqrt{N}$.
\end{lemma}
\begin{proof}
Let $P_{e,j}$ be the error probability when $X_A=\sqrt{\Pu_j}$. We have $P_e=\sum_{j=1}^J (1/J) P_{e,j}$ and
\begin{align}
	P_{e,j} & = \left\{ \begin{array}{ll}
	\Pr\left(Y_A \ge \frac{\sqrt{\Pu_1}+\sqrt{\Pu_{2}}}{2} \right), & j=1 \\
	 \Pr\left(Y_A \le \frac{\sqrt{\Pu_{K-1}}+\sqrt{\Pu_K}}{2} \right), & j=J \\
	 \Pr\left(Y_A \le \frac{\sqrt{\Pu_{j-1}}+\sqrt{\Pu_j}}{2} \right) & \\
	\quad + \Pr\left(Y_A \ge \frac{\sqrt{\Pu_j}+\sqrt{\Pu_{j+1}}}{2} \right), & \text{otherwise}.
	\end{array} \right.
\end{align}
Conditioned on $X_A=\sqrt{\Pu_j}$, $Y_A$ is a Ricean random variable, and hence we compute \cite[p. 50]{ProakisSalehi2008}
\begin{align} \label{eq:interval-error}
	\Pr\left(Y_A \ge \frac{\sqrt{\Pu_j}+\sqrt{\Pu_{j+1}}}{2} \right)
	= Q\left(\frac{\sqrt{\Pu_j}}{\sqrt{N/2}},\frac{\sqrt{\Pu_j}+\sqrt{\Pu_{j+1}}}{2 \sqrt{{N/2}} }\right)
\end{align}
where $Q(a,b)$ is the Marcum Q-function \cite{Corazza2002}.
Consider the following bounds.
\begin{itemize}
\item Upper bound for $b>a$ \cite[UB1MG]{Corazza2002}
\begin{align}
	Q(a,b) \le \exp\left(-\frac{(b-a)^2}{2}\right) .
	\label{eq:marcum_q_ub}
\end{align}

\item Lower bound for $b<a$ \cite[LB2aS]{Corazza2002}
\begin{align}
	& Q(a,b) \nonumber \\
	& \ge 1 - \frac{1}{2} \left[ \exp\left(-\frac{(a-b)^2}{2}\right) - \exp\left(-\frac{(a+b)^2}{2}\right) \right] .
	\label{eq:marcum_q_lb}
\end{align}
\end{itemize}
The bound (\ref{eq:marcum_q_lb}) implies
\begin{align}
	1 - Q(a,b) \le \exp\left(-\frac{(a-b)^2}{2}\right) .
	\label{eq:marcum_q_lb_loosened}
\end{align}

We use \eqref{eq:interval-error} and \eqref{eq:marcum_q_ub}  to write
\begin{align}
	\Pr\left(Y_A \ge \frac{\sqrt{\Pu_j}+\sqrt{\Pu_{j+1}}}{2} \right)\
	\le \exp\left(-\frac{\Delta_{j+1}^2}{4}\right)	.
\end{align}
where $\Delta_j = (\sqrt{\Pu_j} - \sqrt{\Pu_{j-1}})/\sqrt{N}$. Similarly, we use inequality \eqref{eq:marcum_q_lb_loosened} to write
\begin{align}
	\Pr\left(Y_A \le \frac{\sqrt{\Pu_{j-1}}+\sqrt{\Pu_j}}{2} \right)
	& \le \exp\left(-\frac{\Delta_j^2}{4}\right) .
\end{align}
Collecting our results, we have
\begin{align}
	P_e 
	& \leq \frac{1}{J} \left[ 
	\exp\left(-\frac{\Delta_2^2}{4}\right)	
	+ \sum_{j=2}^{J-1} \exp\left(-\frac{\Delta_j^2}{4}\right)			
	\right. \nonumber \\ & \left.
	\qquad + \sum_{j=2}^{J-1} \exp\left(-\frac{\Delta_{j+1}^2}{4}\right)	
	+ \exp\left(-\frac{\Delta_J^2}{4}\right)
	\right] \nonumber \\
	& = \frac{2}{J} \sum_{j=2}^{J} \exp\left(-\frac{\Delta_j^2}{4}\right).
	\label{eq:Pe-bound}
\end{align}

\end{proof}

\section{Orthogonality of Impulse Responses}
\label{sec:ortho_impulses}
We introduce a useful lemma.
\begin{lemma}
\label{lemma:leibniz}
For any complex number $B$, we have
\begin{align}
 \int_{0}^{T_s} \frac{|p(t)|^2}{E_s} e^{B \K(\tau)} d\tau
 	=  \left\{
 	\begin{array}{ll}
 	(e^B-1)/B, 	&	\text{if } B \neq 0	,\\
 	1,					&	\text{if } B   =  0
  \end{array} 
  \right.
\end{align}
where $K(\cdot)$ is defined in (\ref{eq:rise_func}).
\end{lemma}
\begin{proof}
\begin{itemize}
\item 
For $B = 0$, we have
\begin{align}
  \int_{0}^{T_s} \frac{|p(t)|^2}{E_s}  e^{B \K(\tau)} d\tau 
  = \int_{0}^{T_s}  \frac{|p(t)|^2}{E_s}  d\tau
  = 1
\end{align}
where the last equality follows from (\ref{eq:nonzerogvm-pulse-energy}).

\item
For $B \neq 0$, we have
\begin{align}
	\int_{0}^{T_s} \frac{|p(t)|^2}{E_s} e^{B \K(\tau)} d\tau 
	&\stackrel{(a)}{=} \frac{1}{B} \int_{0}^{T_s} B \K^\prime(\tau) e^{B \K(\tau)} d\tau \nonumber\\
	&\stackrel{(b)}{=} \frac{e^{B \K(T_s)} -  e^{B \K(0)}}{B} \nonumber\\
	&\stackrel{(c)}{=} \frac{e^{B} - 1}{B}
\end{align}
where
($a$) follows by applying Leibniz's theorem for differentiation of an integral\cite[3.3.7]{Abramowitz1972}:
\begin{align}
	\K^\prime(t) = \frac{1}{E_s} \frac{d}{d t} \int_0^t |p(\lambda)|^2 d\lambda = \frac{1}{E_s} |p(t)|^2;
\end{align}
($b$) is obtained through integration by substitution and
($c$) holds because $\K(0)=0$ and $\K(T_s)=1$.
\end{itemize}
\end{proof}

Next we show that the impulse responses 
$\{h_{\idx}(\T)\}_{\idx \in \mathcal{F}_1}$ are orthogonal (cf. (\ref{eq:ortho_diff})).
For $\idx_1 \neq \idx_2$, we have
\begin{align}
	\ifdefined\twocolumnmode & \fi
	\int_{-\infty}^{\infty} h_{\idx_1}(\xi) h^*_{\idx_2}(\xi) d\xi
	\ifdefined\twocolumnmode \nonumber\\ \fi
	&\stackrel{(a)}{=} \int_{-\infty}^{\infty} |p(-\xi)|^2 \exp(-i 2\pi (\idx_1 - \idx_2) \K(-\xi)) \ d\xi \nonumber \\
	&\stackrel{(b)}{=} \int_{-\infty}^{\infty} |p(t)|^2 \exp(-i 2\pi (\idx_1 - \idx_2) \K(t)) 	\ dt \nonumber \\
	&\stackrel{(c)}{=} \int_{0}^{\Ts} |p(t)|^2 \exp(-i 2\pi (\idx_1 - \idx_2) \K(t)) 	\ dt \nonumber \\
	&\stackrel{(d)}{=} 0
\end{align}
where 
($a$) follows from $\K^*(\xi) = \K(\xi)$, 
($b$) follows from the transformation of variables $t = -\xi$, 
($c$) holds because $p(t)=0$ for $t \notin [0,\Ts]$ and
($d$) follows from Lemma \ref{lemma:leibniz} with $B=-i 2 \pi (\idx_1 - \idx_2) \neq 0$.

\section{Independence of Filter Outputs}
\label{sec:sufficient}
We drop the user index and time index for notational simplicity.
We decompose $\Y$ into $Y_V$ and
$\overline{\Y}_{V}$ where $\overline{\Y}_{V}$ = $(Y_q: q \in \mathcal{V} \backslash \{V\})$.
Let $\overline{\y} = (\tilde{y}_q: q=1,\ldots,|\mathcal{V}|-1) \in \mathbb{C}^{|\mathcal{V}|-1}$.
The joint pdf of $X$, $Y_V$ and $\overline{\Y}_{V}$ is
\ifdefined\twocolumnmode{
\begin{align}
&p_{X,Y_V,\overline{\Y}_{V}}(x,y,\bar{\y}) \nonumber\\
&= p_{X}(x) \sum_{v \in \mathcal{V}} p_{Y_V,\overline{\Y}_{V},V|X}(y,\bar{\y},v|x) \nonumber \\
&= p_{X}(x) \sum_{v \in \mathcal{V}} p_{V}(v) p_{Z}(y-x) \prod_{q=1}^{|\mathcal{V}|-1} p_{Z}(\tilde{y}_q) \nonumber \\
&= p_{X}(x) ~ p_{Z}(y-x) \prod_{q=1}^{|\mathcal{V}|-1} p_{Z}(\tilde{y}_q) \nonumber \\
&= p_{X}(x) ~ p_{Y_V|X}(y|x) ~ p_{\overline{\Y}_{V}}(\bar{\y}).
\end{align}
}\else{
\begin{align}
p_{X,Y_V,\overline{\Y}_{V}}(x,y,\bar{\y}) 
&= p_{X}(x) \sum_{v \in \mathcal{V}} p_{Y_V,\overline{\Y}_{V},V|X}(y,\bar{\y},v|x) \nonumber \\
&= p_{X}(x) \sum_{v \in \mathcal{V}} p_{V}(v) p_{Z}(y-x) \prod_{q=1}^{|\mathcal{V}|-1} p_{Z}(\tilde{y}_q) \nonumber \\
&= p_{X}(x) ~ p_{Z}(y-x) \prod_{q=1}^{|\mathcal{V}|-1} p_{Z}(\tilde{y}_q) \nonumber \\
&= p_{X}(x) ~ p_{Y_V|X}(y|x) ~ p_{\overline{\Y}_{V}}(\bar{\y}).
\end{align}
}\fi
where $p_{Z}(\cdot)$ is the pdf of $Z_f$ for $f \in \mathcal{V}$ and is given by
\begin{align}
p_Z(z) \stackrel{\Delta}{=} \frac{1}{\pi N} e^{-|z|^2/N}.
\end{align}
Therefore, we find that $X$ and $Y_V$ are statistically independent of $\overline{\Y}_{V}$ and
\begin{align}
I(X ; \overline{\Y}_{V}|Y_V) = 0.
\end{align}
Hence, we have
\begin{align}
I(X; \Y)
&= I(X ; Y_V) + I(X ; \overline{\Y}_{V}|Y_V) 
= I(X ; Y_V).
\end{align}

\section*{Acknowledgement}
We thank the reviewers and the Associate Editor for useful comments.

\appendices

\ifCLASSOPTIONcaptionsoff
  \newpage
\fi

\bibliographystyle{IEEEtran}
\bibliography{bibstrings_short,ieee_dissert4_short}

\begin{IEEEbiography}[{\includegraphics[width=1in,height=1.25in,clip,keepaspectratio]{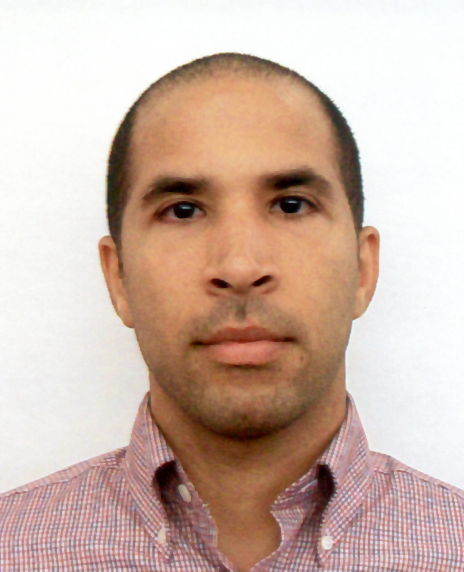}}]{Hassan Ghozlan}
(M'15) received the B.S. degree in electrical engineering from Cairo University, Cairo, Egypt, in 2007, the M.S. degree in electrical engineering from Nile University, Cairo, in 2009, and the Ph.D. degree in electrical engineering from the University of Southern California, Los Angeles, CA, USA.

He was a Research Assistant at the Wireless Intelligent Networks Center at Nile University from 2007 to 2009. He was a visiting Ph.D. student at the Technical University of Munich in 2012 and 2013. He was an intern at Bell Labs in Holmdel, NJ, USA, in 2013. He is currently with Intel Corporation in Hillsboro, OR, USA. His main research interests include communication theory, information theory, and signal processing.
\end{IEEEbiography}

\begin{IEEEbiography}[{\includegraphics[width=1in,height=1.25in,clip,keepaspectratio]{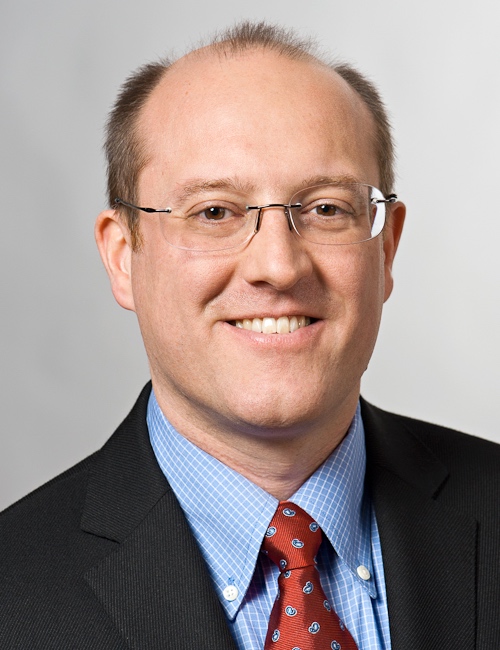}}]{Gerhard Kramer}
(S'91-M'94-SM'08-F'10) received the Dr. sc. techn. degree from ETH Zurich in 1998. From 1998 to 2000, he was with Endora Tech AG in Basel, Switzerland, and from 2000 to 2008 he was with the Math Center at Bell Labs in Murray Hill, NJ, USA. He joined the University of Southern California, Los Angeles, CA, USA, as a Professor of Electrical Engineering in 2009. He joined the Technical University of Munich (TUM) in 2010, where he is currently Alexander von Humboldt Professor and Chair for Communications Engineering. His research interests include information theory and communications theory, with applications to wireless, copper, and optical fiber networks. Dr. Kramer served as the 2013 President of the IEEE Information Theory Society. He was elected to the Bavarian Academy of Sciences and Humanities in 2015.
\end{IEEEbiography}

\end{document}